\documentclass[11pt]{article}
\usepackage{amsthm, amsmath, amssymb, amsfonts, url, bm}
\usepackage{hyperref}
\usepackage{amsthm}
\usepackage{geometry}
\usepackage{hyperref, enumerate}
\usepackage[shortlabels]{enumitem}
\usepackage[capitalise]{cleveref}
\usepackage{algorithm,algpseudocode,algorithmicx}
\usepackage{subcaption}
\usepackage{graphicx}
\usepackage{xcolor}

\counterwithin{figure}{section}

\newtheorem{conjecture}{Conjecture}
\newtheorem{example}{Example}
\newtheorem{theorem}{Theorem}[section]
\newtheorem{prop}[theorem]{Proposition}
\newtheorem{lemma}[theorem]{Lemma}
\newtheorem{cor}[theorem]{Corollary}

\newtheorem{claim}[theorem]{Claim}

\theoremstyle{definition}
\newtheorem{definition}[theorem]{Definition}
\newtheorem*{defn-non}{Definition}

\newtheorem{rmk}[theorem]{Remark}

\newlist{Case}{enumerate}{3}
\setlist[Case, 1]{%
    label           =   {\bfseries Case \arabic*.},
    labelindent=1em ,labelwidth=1cm, labelsep*=1em, leftmargin =!
}
\setlist[Case, 2]{%
    label           =   {\bfseries Subcase \arabic{Casei}.\arabic*.},
    labelindent=-1em ,labelwidth=1cm, labelsep*=1em, leftmargin =!
}
\setlist[Case, 3]{%
    label           =   {\bfseries Subsubcase \arabic{Casei}.\arabic{Caseii}.\arabic*.},
    labelindent=-1em ,labelwidth=1cm, labelsep*=1em, leftmargin =!
}

\newenvironment{poc}{\begin{proof}[Proof of claim]}{\end{proof}}

\newcommand{\cA}{\mathcal{A}}

\newcommand{\cC}{\mathcal{C}}

\newcommand{\cF}{\mathcal{F}}
\newcommand{\cG}{\mathcal{G}}
\newcommand{\cH}{\mathcal{H}}
\newcommand{\cI}{\mathcal{I}}
\newcommand{\cJ}{\mathcal{J}}

\newcommand{\cL}{\mathcal{L}}

\newcommand{\cS}{\mathcal{S}}

\newcommand{\cZ}{\mathcal{Z}}

\renewcommand{\le}{\leqslant}
\renewcommand{\leq}{\leqslant}
\renewcommand{\ge}{\geqslant}
\renewcommand{\geq}{\geqslant}

\usepackage{mathtools}
\DeclarePairedDelimiter\abs{\lvert}{\rvert}

\DeclarePairedDelimiter\ceil{\lceil}{\rceil}
\DeclarePairedDelimiter\floor{\lfloor}{\rfloor}
\DeclarePairedDelimiter\parenv{\lparen}{\rparen}

\DeclarePairedDelimiter\set{\{}{\}}
\DeclarePairedDelimiter\mset{\{\!\!\{}{\}\!\!\}}

\DeclareMathOperator{\supp}{supp}
\DeclareMathOperator{\wt}{wt}
\DeclareMathOperator{\asy}{asy}
\DeclareMathOperator{\Com}{Com}
\DeclareMathOperator{\Ord}{Ord}
\DeclareMathOperator{\gd}{d}
\DeclareMathOperator{\cod}{cod}
\DeclareMathOperator{\jd}{jd}
\DeclareMathOperator{\gdmax}{d_{\max}}
\DeclareMathOperator{\codmax}{cod_{\max}}
\DeclareMathOperator{\jdmax}{jd_{\max}}
\DeclareMathOperator{\Maj}{Maj}
\DeclareMathOperator*{\argmax}{arg\,max}
\newcommand{\ra}{r^*_{\asy}}
\newcommand{\La}{L_{\asy}}

\newcommand{\N}{\mathbb{N}}

\newcommand{\Z}{\mathbb{Z}}

\newcommand{\ba}{\boldsymbol{a}}
\newcommand{\bb}{\boldsymbol{b}}
\newcommand{\bc}{\boldsymbol{c}}
\newcommand{\be}{\boldsymbol{e}}

\newcommand{\bu}{\boldsymbol{u}}
\newcommand{\bv}{\boldsymbol{v}}
\newcommand{\bw}{\boldsymbol{w}}
\newcommand{\bx}{\boldsymbol{x}}

\newcommand{\by}{\boldsymbol{y}}
\newcommand{\bz}{\boldsymbol{z}}
\newcommand{\balpha}{\boldsymbol{\alpha}}
\newcommand{\bbeta}{\boldsymbol{\beta}}
\newcommand{\bgamma}{\boldsymbol{\gamma}}
\newcommand{\ch}{\mathfrak{ch}}

\DeclareMathOperator{\Ball}{Ball}
\algrenewcommand\algorithmicrequire{\textbf{Input:}}
\algrenewcommand\algorithmicensure{\textbf{Output:}}
\algnewcommand{\LineComment}[1]{\State \(\triangleright\) #1}
 
\title{Optimal Reconstruction Codes with Given Reads in Multiple Burst-Substitutions Channels}
\author{
Wenjun Yu\thanks{School
of Electrical and Computer Engineering, Ben-Gurion University of the Negev,
Beer Sheva 8410501, Israel
(e-mail: wenjun@post.bgu.ac.il).}
\and 
Yubo Sun\thanks{School of Mathematical Sciences, Capital Normal University, Beijing 100048, China (e-mail: 2200502135@cnu.edu.cn, gnge@zju.edu.cn). The research of G. Ge was supported by the National Key Research and Development Program of China under Grant 2020YFA0712100, the National Natural Science Foundation of China under Grant 12231014, and Beijing Scholars Program.}
\and
Zixiang Xu\thanks{Extremal Combinatorics and Probability Group (ECOPRO), Institute for Basic Science (IBS), Daejeon 34126, South Korea.
(e-mail: zixiangxu@ibs.re.kr). The research of Z. Xu was supported by IBS-R029-C4.}
\and
Gennian Ge\footnotemark[2]
\and
Moshe~Schwartz\thanks{Department of Electrical and Computer Engineering, McMaster University, Hamilton, ON, L8S 4K1, Canada, and on a leave of absence from the School
of Electrical and Computer Engineering, Ben-Gurion University of the Negev,
Beer Sheva 8410501, Israel
(e-mail: schwartz.moshe@mcmaster.ca).}
}

\begin{document}
\date{}
\maketitle
\begin{abstract}
We study optimal reconstruction codes over the multiple-burst substitution channel. Our main contribution is establishing a trade-off between the error-correction capability of the code, the number of reads used in the reconstruction process, and the decoding list size. We show that over a channel that introduces at most $t$ bursts, we can use a length-$n$ code capable of correcting $\epsilon$ errors, with $\Theta(n^\rho)$ reads, and decoding with a list of size $O(n^\lambda)$, where $t-1=\epsilon+\rho+\lambda$. In the process of proving this, we establish sharp asymptotic bounds on the size of error balls in the burst metric. More precisely, we prove a Johnson-type lower bound via Kahn's Theorem on large matchings in hypergraphs, and an upper bound via a novel variant of Kleitman's Theorem under the burst metric, which might be of independent interest.

Beyond this main trade-off, we derive several related results using a variety of combinatorial techniques. In particular, along with tools from recent advances in discrete geometry, we improve the classical Gilbert-Varshamov bound in the asymptotic regime for multiple bursts, and determine the minimum redundancy required for reconstruction codes with polynomially many reads. We also propose an efficient list-reconstruction algorithm that achieves the above guarantees, based on a majority-with-threshold decoding scheme.
\end{abstract}

\section{Introduction}
\subsection{Overview and related work}

The reconstruction problem introduced by Levenshtein in~\cite{levenshtein2001efficient} and~\cite{levenshtein2001efficientjcta}, involves transmitting a codeword through multiple identical noisy channels, resulting in a set of distinct outputs. The goal of the reconstruction process is to recover the original codeword from these outputs. Formally, we are given a channel through which we transmit a codeword $\bx$. This may result in a corrupted version of $\bx$, the set of all of which is called the error ball of $\bx$. In the reconstruction scheme, a codeword $\bx$ is transmitted through $N$ noisy channels with the same error model, $N$ distinct outputs (called reads) within the error ball of $\bx$ are received. This general reconstruction problem extends the classic error-correction problem, which corresponds to the special case where $N = 1$. This problem is particularly relevant in fields such as molecular biology, chemistry and advanced memory storage technologies including associative memories~\cite{yaakobi2012uncertainty}, racetrack memories~\cite{chee2018reconstruction}, and DNA storage~\cite{abu2021levenshtein}.

Levenshtein~\cite{levenshtein2001efficient} was the first to study the substitution, insertion, and deletion channels. The work in~\cite{abu2021levenshtein} extended this by investigating the case where errors are combinations of a single substitution and a single insertion. The channel involving single edits (i.e., a single substitution, insertion, or deletion), and its variants, have been explored in~\cite{cai2021coding}. Other types of channels, such as those involving tandem duplications~\cite{yehezkeally2019reconstruction}, limited-magnitude errors~\cite{wei2022sequence}, and single bursts of edits and their variants~\cite{sun2023sequence}, have also been studied. In particular, the trade-off between codebook redundancy and the number of reads required have been discussed in~\cite{cai2021coding}. Specifically, the authors in~\cite{cai2021coding} focused on the design of optimal reconstruction codes with minimal redundancy (equivalently, maximal number of codewords), given the number of reads. The work of~\cite{cai2021coding} also presented reconstruction codes with asymptotically optimal redundancy in the single-edit channel. For channels involving two insertions or deletions, the best reconstruction codes with a given number of reads can be found in~\cite{sabary2024survey,Sun2023reconstructions,Ye2023reconstructions}.

Previous research has primarily focused on non-burst errors or a single burst error. Here, a single burst-substitution error of length $b$ is an occurrence of substitution errors that are confined to $b$ consecutive positions in the codeword. In this paper, we are interested in the channel involving multiple burst errors. This type of error is relevant in various applications, including high-density magnetic recording channels~\cite{kuznetsov1993coding,levenshtein1993perfect} and DNA-based storage systems~\cite{jain2020coding}. For the single burst-substitution channel, Abdel-Ghaffar \textit{et al.}~\cite{abdel1986existence,abdel1988existence} and Etzion~\cite{etzion2001constructions} studied perfect correcting codes under specific parameters. More recently, Wei and Schwartz \cite{wei2022perfect} investigated single bursts with limited magnitude, and Nguyen \textit{et al.} \cite{nguyen2022optimal} studied single inversions, both are specific types of burst substitutions.
For the multiple burst-substitutions channel, Overveld~\cite{van1987multiple}, Blaum \textit{et al.}~\cite{BlaFarTil88}, and Chen \textit{et al.}~\cite{zhi1992constructions}, have used product codes and disjoint difference sets to construct codes capable of correcting multiple bursts under certain parameters. The redundancies of known constructions of error-correcting codes are relatively large.

The authors in~\cite{sabary2024survey} demonstrated that the optimal redundancies of reconstruction codes in the channel with two insertions or deletions can be classified into two classes based on the number of reads $N$. The redundancy decreases to zero as $N$ increases linearly with the codeword length $n$. Generally, as the number of errors increases, the optimal redundancies are divided into more classes. Therefore, we investigate the asymptotic behavior of the redundancy of reconstruction codes given a specific number of reads. Specifically, we shall explore the optimal redundancy of reconstruction codes within $\Theta(n^s)$ reads as the codeword length $n$ approaches infinity. Our focus is on how the order of reads $\Theta(n^s)$ impacts the best achievable redundancy.

Another issue we study is whether we can gain by moving from unique decoding to list decoding. The classical reconstruction problem focuses on recovering a unique codeword from multiple reads. The list-reconstruction problem, on the other hand, relaxes this requirement by allowing recovery of a bounded list of possible codewords. Specifically, a codeword $\bx$ is transmitted through $N$ noisy channels, producing $N$ distinct reads within the error ball of $\bx$. The list-reconstruction problem identifies all codewords $\by$ such that all $N$ reads are contained within the error ball of $\by$. 

In this paper, we combine all of these aspects together, noting the main two trade-offs possible: First, while the channel may introduce $t$ burst errors of length $b$, we may opt to trade-off code power for number of reads, and use a code designed to correct fewer burst errors, but ask for multiple reads. Second, we may trade-off unique decoding for reads, namely, reduce the number of reads required by the decoding algorithm, trading this reduction for an increase in the list size produced by the decoder.

Noted previous works on list reconstruction includes Junnila \emph{et al.}~\cite{junnila2021levenshtein}, in which bounds on the list size with reconstruction over the substitution channel were studied. In a follow up work~\cite{junnila2023levenshtein}, the same  authors investigated the maximum list size of error-correcting codes in the substitution channel and determined the read number for small list sizes. Finally, Yehezkeally and Schwartz~\cite{yehezkeally2021reconstruction} studied list-reconstruction over the uniform-tandem-duplication channel. We observe that no study thus far has tackled list reconstruction over a bursty channel.

\subsection{Our contributions}

Throughout the paper we focus on the channel $\ch(t,b)$, which may modify any transmitted vector by substituting at most $t$ bursts of length at most $b$. We note that this channel includes as special cases the standard multiple-substitution channel $\ch(t,1)$, and the single-burst channel $\ch(1,b)$, which have been studied extensively in the past. Thus, our results regarding list reconstruction apply to them as well.

Our first contribution concerns optimal error-correcting codes over $\ch(t,b)$. Here, for a code $\cC$, over an alphabet of size $q$, with codewords of length $n$, and a total of $M$ codewords, \emph{optimal} refers to the code having the largest $M$ of all codes capable of correcting errors induced by $\ch(t,b)$. The redundancy of the optimal code, $r(\cC)=n-\log_q M$, is therefore the smallest of all codes for the channel. We show in Corollary~\ref{cor:Redu}, that any optimal burst-correcting code for $\ch(t,b)$ satisfies
\[
(1-o(1))t\log_q n \leq r(\cC)\leq (1+o(1))2t\log_q n.
\]
The asymptotic regime studied throughout the paper assumes the alphabet size $q$, the number of burst errors $t$, and their maximal length $b$, are all fixed, while the code length $n\to\infty$. Through an improvement of the Gilbert-Varshamov bound, we give a stronger upper bound in Theorem~\ref{thm:impoved upperred},
\[
r(\cC) \leq 2t\log_q n - \log_q \log n + O(1).
\]
Interestingly, the bounds hint at the fact that the dominant term of the redundancy of the optimal error-correcting codes does not depend on the burst length $b$, but rather on the number of bursts. This is in particular interesting, since the common practical approach for correcting multiple bursts is by using interleaving of single-burst-correcting codes, or by concatenated codes and extension fields. Both of these approaches result in redundancy $(1+o(1))bt\log n$, which is worse than the guarantee of Corollary~\ref{cor:Redu}.

We then move to study the first trade-off: reducing the redundancy of the code by increasing the number of required reads (and thereby, moving from an error-correcting code to a reconstruction code). To that end, we say the number of reads $N(n)$ (which depends on the length of the code) has $b$-order of $s$ if $N(n)=\Omega(n^s)$, $N(n)=o(n^{s+1})$, and $N(n)$ is strictly larger than the error-ball of radius $s$ (see Definition~\ref{def:ord}). Then, in Theorem~\ref{thm:MARR for burst} we show that if $\cC'$ is an optimal reconstruction code for $\ch(t,b)$ with $N(n)$ of $b$-order $s$, and $\cC$ is an optimal error-correcting code for $\ch(t-s-1,b)$, then
\[
r(\cC') = (1+o(1)) r(\cC).
\]
From a theoretical point of view, this theorem shows how we can trade error-correcting power (hence, reducing redundancy) for number of reads. Roughly speaking, any reduction of one error-correction capability increases the polynomial degree of the number of required reads by one. From a practical point of view, our knowledge of multiple-burst-correcting codes is severely limited, with only a handful of constructions known for restrictive sets of parameters. In contrast, there is quite a wealth of knowledge concerning single-burst-correcting codes. With our result we can now use these known codes, designed to correct only a single burst, to correct multiple bursts when given polynomially-many reads.

The final trade-off we consider continues the previous one: we have the channel $\ch(t,b)$, and an error-correcting code, $\cC$, capable of correcting only $t-s-1$ burst errors of length at most $b$. While using $\cC$ over this channel, $\Theta(n^s)$ reads suffice for reconstruction, we further lower the number of reads to $\Theta(n^{s-h})$. In Theorem~\ref{Theorem:Main 1:List} we provide the main result, showing that in this case we get a list-decoder with list size $O(n^h)$. Roughly summarizing the complete trade-off between all parameters involved, we show that over a channel that introduces at most $t$ bursts, we can use a length-$n$ code capable of correcting $\epsilon$ errors, with $\Theta(n^\rho)$ reads, and decoding with a list of size $O(n^\lambda)$, where
\[t-1=\epsilon+\rho+\lambda.\]
This allows us to conveniently distribute our ``budget'' of $t-1$ between error-correcting capability, number of reads, and decoder list size. Once again, we note that the dominant asymptotic term here does not depend on the (constant) burst length $b$. In addition, we describe a polynomial-time list-reconstruction algorithm which only assumes a list-decoder for the code (i.e., using a single read), over $\ch(t-s+h,b)$.

\subsection{Techniques}

We use a variety of combinatorial techniques as well as coding-theoretic ones. We highlight the important ones.

When studying the burst metric, the most fundamental result we develop is a sufficiently sharp asymptotic estimate of the size of a ball with constant radius. We further extend this to give good estimates on the size of intersection of two balls. This allows us to prove the results on the redundancy of optimal burst-correcting codes and optimal burst-reconstruction codes. Together with a recent geometric result~\cite{campos2023new}, we can show an improvement to the classical Gilbert-Varshamov-like bound for the burst metric.

The main result of Theorem~\ref{Theorem:Main 1:List}, providing the complete trade-off between the error-correction capability, number of reads, and list size, is proved using separate upper and lower bounds. To prove the bounds, we employ various combinatorial arguments that may be of independent interest. To establish the upper bound, we prove a novel variant of the celebrated Kleitman's Theorem~\cite{1966Kleitman} under the burst distance using the shift argument. It shows that the largest set of vectors with bounded burst-distance diameter must be a ball in the burst metric. Conversely, we leverage Kahn's result~\cite{kahn1996linear} on finding large matchings to derive a Johnson-type bound, which bounds the number of codewords found within a ball of a given radius in the burst metric.

The list-reconstruction algorithm we propose uses a majority-with-threshold technique. All reads are considered to create a point-wise majority with a threshold. When the threshold is not achieved, a ``joker'' symbol is used. The threshold is carefully chosen to make sure only a constant number of jokers is generated. By exhaustively replacing jokers with all possible assignments from the alphabet and employing a list decoder over a weaker channel, a list is created, which we prove is the correct list-reconstruction.

\subsection{Paper organization}

The paper is organized as follows: Section~\ref{sec:prelim} introduces the relevant notations, definitions and tools. In Section~\ref{sec:ecc}, we study the size of error balls and the redundancy of burst-correcting codes. Section~\ref{sec:reconstruction code} presents the optimal redundancy of reconstruction codes with $\Theta(n^s)$ reads. Section~\ref{sec:list} obtains the list size of reconstructing codewords based on burst-correcting codes with fewer reads ($\Theta(n^{s-h})$). Section~\ref{sec:alg} provides an algorithm for list reconstructing codewords based on burst-correcting codes. Finally, Section~\ref{sec:con} concludes the paper with a summary and open questions.

\section{Preliminaries}\label{sec:prelim}
\subsection{Notations}
A \emph{hypergraph} \( \cH \) is a pair \( (V(\cH), E(\cH)) \), where \( V(\cH) \) (called the \emph{vertex set}) is a non-empty set, and \( E(\cH) \) (called the \emph{edge set}) is a collection of non-empty subsets of \( V(\cH) \). The hypergraph \( \cH \) is said to be \emph{\( r \)-uniform} if \( \abs{e} = r \) for every \( e \in E(\cH) \). In particular, a \( 2 \)-uniform hypergraph is also known as a \emph{graph}. Two vertices \( x, y \in V(\cH) \) are \emph{adjacent} in \( \cH \) if there exists an edge \( e \in E(\cH) \) such that \( x, y \in e \). For any two distinct vertices \( u, v \in V(\cH) \), the \emph{degree} of \( u \) is defined as  
\[
\gd(u) = \abs*{\set*{ e \in E(\cH) : u \in e }}.
\]  
Similarly, the \emph{codegree} of \( u \) and \( v \) is given by  
\[
\cod(u,v) = \abs*{\set*{ e \in E(\cH) : u, v \in e }}.
\]
The \emph{joint degree} of $u$ and $v$ counts the number of vertices that are simultaneously neighbors of $u$ and $v$,
\[
\jd(u,v)=\abs*{\set*{w\in V(\cH): \exists e,e'\in E(\cH) \text{ s.t. } \set{u,w}\subseteq e, \set{v,w}\subseteq e'}}.
\]
The \emph{maximum degree}, \emph{maximum codegree}, and \emph{maximum joint degree} of \( \cH \) are denoted by  
\begin{align*}
\gdmax(\cH) &= \max_{u \in V(\cH)} \gd(u), &
\codmax(\cH) &= \max_{u \neq v \in V(\cH)} \cod(u,v), &
\jdmax(\cH) &= \max_{u \neq v \in V(\cH)} \jd(u,v),
\end{align*}
respectively. An \emph{independent set} of \( \cH \) is a subset \( S \subseteq V(\cH) \) such that no edge of \( \cH \) is entirely contained in \( S \). The \emph{independence number} \( \alpha(\cH) \) is defined as the size of the largest independent set in \( \cH \). A \emph{matching} \( M \subseteq E(\cH) \) is a collection of edges such that no two edges share a common vertex. The \emph{matching number} \( \nu(\cH) \) is the size of the largest matching in \( \cH \).

For any \( i, j \in \Z \) with \( i \leq j \), define the interval  
\[
[i, j] = \set*{i, i+1, \dots, j}.
\]  
When considering cyclic structures, we extend this definition to the \emph{cyclic interval} \( [i, j]_C \), which coincides with \( [i, j] \) if \( i \leq j \) and is defined as \( [i, n] \cup [1, j] \) otherwise. For simplicity, we typically omit the subscript and use \( [i, j] \) for cyclic intervals unless explicitly stated otherwise.

Given \( n \in \N \), we denote \( [n] = [1, n] \), and define \( \overline{[i, j]} \) as the complement of \( [i, j] \) within \( [n] \), where $n$ is implicitly understood from the context. For a cyclic interval \( I = [i, j] \), its \emph{\( (k, b) \)-extension} is given by  
\begin{equation}
\label{eq:ext}
[i - kb, j + kb].
\end{equation}
The \emph{length} of a cyclic interval \( [i, j] \), denoted by \( \abs{[i, j]} \), is defined as  
\[
j - i + 1 \quad \text{if } i \leq j, \quad \text{and} \quad n + j - i + 1 \quad \text{otherwise}.
\]  
The \emph{gap} between two cyclic intervals \( [i_1, j_1] \) and \( [i_2, j_2] \) is given by  
\[
\min\set*{ i_2 - j_1 - 1, i_1 + n - j_2 - 1 }.
\]
Two cyclic intervals intersect if and only if their gap is negative.

Throughout the paper we shall work over a finite alphabet $\Sigma$. We assume it forms an Abelian group of size $q$, which we may emphasize by writing $\Sigma_q$. An \( n \)-tuple \( \bx = (x_1, x_2, \dots, x_n) \), where each \( x_i \in \Sigma \), is referred to as a \emph{vector} or \emph{codeword} of length \( n \) over \( \Sigma \). The set of all sequences of length \( n \) over \( \Sigma \) is denoted by \( \Sigma^n \). We use \( \bx^b \) to denote the concatenation of \( b \) copies of \( \bx \). The sequence consisting entirely of zeros is denoted by \( \mathbf{0} \) (where the length is implicitly understood from the context).

For indices \( i \leq j \), the consecutive symbols \( x_i, \dots, x_j \) are denoted by \( \bx[i, j] \), and the \( i \)-th symbol of \( \bx \) is written as \( \bx[i] \). Similarly, $\bx\overline{[i,j]}$ denotes the symbols $x_1,\dots,x_{i-1},x_{j+1},\dots,x_n$. The \emph{length} of \( \bx\in\Sigma^n \) is given by \( \abs{\bx}=n \), and its \emph{support} is defined as  
\[
\supp(\bx) = \set*{ i \in [n] : \bx[i] \neq 0 }.
\]  
For any two sequences \( \bx, \by \in \Sigma^n \), the support trivially satisfies  
\[
\supp(\bx + \by), \supp(\bx - \by) \subseteq \supp(\bx) \cup \supp(\by).
\]  
Finally, given a positive integer \( t \) and a set \( S \), the notation \(\binom{S}{t} \) denotes the collection of all subsets of \( S \) containing exactly \( t \) distinct elements.

\subsection{Tools from graph theory}
Finding the matching number $\nu(\cH)$ of a hypergraph $\cH$ is indeed a fundamental and challenging problem in combinatorics and optimization, which can help understanding the structure of hypergraphs. It is related to other problems like vertex covers, independent sets, and edge colorings. Since finding the exact matching number is NP-hard for general hypergraphs, we often focus on theoretical bounds. For example, Kahn's Theorem~\cite[Theorem 1.2]{kahn1996linear} provides an approach to finding large matchings in certain hypergraphs, which was generalized to the following version:

 \begin{lemma}[{{\cite[Lemma 2.4]{liu2025approximate}}}]
 \label{lem:kahn}
For fixed positive integer $r\ge 2$, let $\cH$ be an $r$-uniform hypergraph with $\frac{\codmax(\cH)}{\gdmax(\cH)} = o(1)$, where $o(1) \to 0$ as $\abs{V(\cH)} \to \infty$. Then we have 
\[
\nu(\cH) \geq (1-o(1)) \frac{\abs*{E(\cH)}}{\gdmax(\cH)}.
\]
 \end{lemma}

 We will take advantage of a recent breakthrough in discrete geometry. Specifically, Campos, Jenssen, Michelen, and Sahasrabudhe~\cite{campos2023new} have advanced a longstanding lower bound for sphere packing in high-dimensional spaces. A central idea from their work, encapsulated in the following lemma, enables us to derive an improved Gilbert–Varshamov bound in our context.

\begin{theorem}[{{\cite[Theorem 1.3]{campos2023new}}}]
\label{codegree}
Let $G$ be a graph on $n$ vertices. If the maximum degree and maximum joint degree of $G$ satisfy that $\gdmax(G) \leq \Delta$ and $\jdmax(G) \leq \frac{\Delta}{(2\log \Delta)^7}$, then
\[
    \alpha(G) \geq (1-o(1))\frac{n\log \Delta}{\Delta},
\]
where $o(1)\to 0$ as $\Delta \to \infty$.
\end{theorem}


\subsection{Burst metric and channel}

Building on previous studies on bursts of substitutions (e.g.,~\cite{abdel1986existence,abdel1988existence,etzion2001constructions}), bursts can be categorized as either cyclic or non-cyclic. Since the results we can obtain are similar in both scenarios, we choose to focus on cyclic burst substitutions.

A $b^{\leq}$-burst error is a vector $\be\in\Sigma_q^n$ whose support, $\supp(\be)$, is contained in a cyclic interval of length at most $b$. Let $I(\be)$ denote the smallest cyclic interval containing $\supp(\be)$. Two $b^{\leq}$-burst errors, $\be,\be'\in\Sigma_q^n$ are said to be disjoint if $I(\be)\cap I(\be')=\emptyset$.

For any non-negative integer $t$, define $B_{t,\leq b}^\star(q, n)$ as the set of sequences $\bw \in \Sigma_q^n$ that can be expressed as a sum of $t$ $b^{\leq}$-bursts, i.e., $\bw = \be_1 + \cdots + \be_t$, where each $\be_i$ is a $b^{\leq}$-burst. Let $B_{t,\leq b}(q, n)$ denote the subset of $B_{t,\leq b}^\star(q, n)$ where we further require that in some decomposition of $\bw=\be_1+\dots+\be_t$, we have that $\be_1,\dots,\be_t$ are pairwise disjoint $b^{\leq}$-burst errors. Note that $B_{0,\leq b}^\star(n,q) = B_{0,\leq b}(n,q) = \set{\mathbf{0}}$. If $q$ or $n$ are clear from context, we may omit them from our notations, e.g., writing $B_{t,\leq b}$ instead of $B_{t,\leq b}(n,q)$.

 \begin{example}
   Let $\bw = 01011 \in \Sigma_2^5$ be a sequence with $\supp(\bw) = \set{2, 4, 5}$. The sequence $\bw$ can be expressed in various ways, for example,
\begin{align*}
\bw &= 01100 + 00110 + 00001\in B_{3, \leq 2}^\star(5,2), \\ 
\bw &= 01000 + 00010 + 00001 \in B_{3, \leq 2}(5,2), \\ 
\bw &= 01000 + 00011 \in B_{2, \leq 2}(5,2).
\end{align*}
In the last expression, $\bw$ is represented by two disjoint $2^{\leq}$-bursts: $\be_1 = 01000$ and $\be_2 = 00011$, where the smallest cyclic intervals containing their supports are $I(\be_1) = [2, 2]$ and $I(\be_2) = [4, 5]$. The gap between $I(\be_1)$ and $I(\be_2)$ is $1$.
\end{example}

\begin{lemma}\label{lem:disjoint}
    For any integers $n,t,b,q\geq 2$, with $n \geq tb$, we have $B_{t,\leq b}^\star \subseteq \bigcup_{i=0}^t B_{i,\leq b}$. Consequently, $    
    \bigcup_{i=0}^t B_{i,\leq b}^\star  = \bigcup_{i=0}^t B_{i,\leq b}.$
\end{lemma}

\begin{proof}[Proof of~\cref{lem:disjoint}]
    We prove the claim by induction on $t$. For the base cases, it is clear that $B_{0,\leq b}^\star = B_{0,\leq b}=\set{\mathbf{0}}$ and $B_{1,\leq b}^\star = B_{1,\leq b}$. Now assume that the result holds for all integers less than $t$. Let $\bw = \be_1 + \cdots + \be_t \in B_{t,\leq b}^\star$. By the induction hypothesis, we know that $\be_1 + \cdots + \be_{t-1} \in \bigcup_{i=0}^{t-1} B_{i,\leq b}$ can be represented by $s$ disjoint $b^{\leq}$-bursts, where $s \leq t - 1$. That is, there exist $b^{\leq}$-bursts $\set{\bu_{i}}_{i=1}^{s}$ such that $\be_1 + \cdots + \be_{t-1} = \bu_1 +\cdots + \bu_s$. Moreover, the sets in $\set{I_i = I(\bu_i)}_{i=1}^s$ are pairwise disjoint and $\abs{I_i} \leq b$ for $i \in [s]$.

    Define $I_t = I(\be_t)$. We now consider two cases based on the relationship between $I_t$ and $\set{I_1, \ldots, I_s}$.

    \begin{enumerate}
        \item If there exists an interval $I_t \subseteq I_j$ for some $j \in [s]$, then $\bu_j + \be_t$ forms a single $b^{\leq}$-burst with $I(\bu_j+\be_t)\subseteq I_j$. In this case, we have
        \[
    \bw = \bu_1 +\cdots + \bu_{j-1} + \bu_{j+1} + \cdots + \bu_s + (\bu_j + \be_t) \in  \bigcup_{i=0}^t B_{i,\leq b}.
    \]
    \item If $I_t \not\subseteq I_j$ for all $j \in [s]$, assume w.l.o.g., that $I_1,\dots,I_{s'}$, intersect $I_t$, $s\leq s'$. Since $I_1,\dots,I_s$ are pairwise disjoint, out of $I_1,\dots,I_{s'}$, we have $m$ that are not fully contained in $I_t$, $m\leq 2$.
    \begin{enumerate}
    \item If $m=0$, then $\bu_1+\dots+\bu_{s'}+\be_t$ is a $b^{\leq}$-burst, and $I(\bu_1+\dots+\bu_{s'}+\be_t),I_{s'+1},\dots,I_s$ are pairwise disjoint, hence $\bw\in\bigcup_{i=0}^t B_{i,\leq b}$.
    \item If $m = 1$, suppose the unique interval that intersects $I_t$, but is not contained in it, is $I_{s'}$. Then, $I_t \cup I_{s'}, I_{s'+1}, \dots, I_s$ are pairwise disjoint. However, $\bu_{s'}+\be_t$ is not necessarily a $b^{\leq}$-burst, but rather a $(2b-1)^{\leq}$-burst. If it is a $b^{\leq}$-burst, we are done. Otherwise, we can easily partition $I_t\cup I_{s'}$ into two disjoint cyclic intervals of size at most $b$, $I'_1$ and $I'_2$, and present $\bu_{s'}+\be_t=\be'_1+\be'_2$ where $\be'_1$ and $\be'_2$ are both $b^{\leq}$-bursts with $I(\be'_1)\subseteq I'_1$, $I(\be'_2)\subseteq I'_2$. It then follows that $\bw\in\bigcup_{i=0}^t B_{i,\leq b}$.
    \item If $m = 2$, the same argument as in b) applies. We omit the repeated details.
    \end{enumerate}
    \end{enumerate}
    This completes the proof.
\end{proof}

For simplicity, we define $B_{\leq t,\leq b} = \bigcup_{i=0}^t B_{i,\leq b}^{\star}= \bigcup_{i=0}^t B_{i,\leq b}$, where the last equality is due to Lemma~\ref{lem:disjoint}. The set $B_{\leq t,\leq b}$ contains all the error patterns resulting from at most $t$ $b^{\leq}$-bursts. Since these are exactly the errors introduced by the channel we study, we have the following definition.

\begin{definition}
\label{def:ch}
 Define $\ch(t, b)$ as the channel allowing up to $t$ bursts, each of length at most $b$. In particular, when $b = 1$, the channel has at most $t$ single-symbol substitutions, in which case we omit $b$ from the notation and write $\ch(t)$ for simplicity.
\end{definition}

The \emph{error ball} of radius $t$ around a sequence $\bx\in\Sigma_q^n$, denoted $\Ball_{t, b}(\bx)$, consists of all sequences that can be obtained from $\bx$ by applying at most $t$ $b^{\leq}$-bursts. Formally, we define
\[\Ball_{t,b}(\bx) = \set*{\bx + \be: \be \in B_{\leq t,\leq b}}.\]
When considering up to $t$ occurrences of $b^{\leq}$-bursts, we may assume by Lemma~\ref{lem:disjoint} that these bursts are pairwise disjoint. Thus, the channel $\ch(t,b)$ maps any incoming sequence $\bx$ to a possibly corrupted output from $\Ball_{t,b}(\bx)$. We can now define a burst-correcting code formally.

\begin{definition}
A \emph{$(t,b)$-burst-correcting code} is a set $\cC \subseteq \Sigma_q^n$ such that
\[
\Ball_{t,b}(\bu) \cap \Ball_{t,b}(\bv) = \emptyset,
\]
for any two distinct $\bu,\bv \in \cC$. The \emph{rate} of the code $\cC$ is defined as
\[R(\cC)= \frac{1}{n}\log_q\abs{\cC},\]
and the \emph{redundancy} of $\cC$ is given by 
\[r(\cC) = (1 - R(\cC))n.\]
\end{definition}

For any $\bx,\by \in \Sigma_q^n$, the \emph{Hamming distance} $d_H(\bx,\by)$ is the minimum number of substitutions required to transform $\bx$ into $\by$, i.e.,
\[
d_H(\bx,\by) = \abs*{\set*{i\in[n] : \bx[i]\neq \by[i]}}.
\]
The Hamming weight of $\bx$ is defined as $\wt_H(\bx)=d_H(\bx,\mathbf{0})$.

We define the \emph{$b^{\leq}$-burst distance}, $d_b(\bx, \by)$, as the minimum number of $b^{\leq}$-bursts required to transform $\bx$ into $\by$. Formally,
\[
d_{b}(\bx,\by) = \min \set*{t: \bx - \by \in B_{\leq t,\leq b}}=\min\set{ t : \by\in\Ball_{t,b}(\bx)}.
\]
It is straightforward to verify that the $b^{\leq}$-burst distance is a valid metric on $\Sigma_q^n$. The \emph{weight} of $\bx$ with respect to $b^{\leq}$-burst distance, denoted $\wt_b(\bx)$, is defined as $\wt_b(\bx) = d_b(\bx, \mathbf{0})$. For a collection of vectors $\cA$, the \emph{diameter} of $\cA$ with respect to the $b^{\leq}$-burst distance is defined as
\[
D(\cA) = \max_{\bx\neq \by \in \cA} d_b(\bx,\by).
\]

Due to $d_b$ being a metric, we make the standard observation that $\cC\subseteq\Sigma_q^n$ is a $(t,b)$-burst-correcting code, if and only if, for any distinct $\bx,\by\in\cC$, we have $d_b(\bx,\by)\geq 2t+1$.

\subsection{The reconstruction problem}

We now introduce the \emph{reconstruction problem} in the context of a general noisy channel $\ch$, where each sequence $\bx$ has an associated error ball, $\Ball(\bx)$. In this model, a sequence $\bx \in \cC \subseteq \Sigma_q^n$ is transmitted multiple times through the channel, with each transmission resulting in a sequence from $\Ball(\bx)$. Suppose we obtain $N$ distinct reads, $\by_1, \ldots, \by_{N} \in \Ball(\bx)$. Using these $N$ reads, we aim to uniquely decode the original sequence $\bx$. A code $\cC$ allowing this is called an $(n,\ch,N)_q$ reconstruction code. We refer to the minimal integer $N$ allowing unique decoding as $N(\cC)$, called the \emph{reconstruction degree}. Obviously, if $\cC$ is a $(n,\ch,N(\cC))_q$ reconstruction code, it is also an $(n,\ch,N)_q$ reconstruction code for all $N\geq N(\cC)$. It is clear that successful unique decoding implies:
\begin{equation}
\label{eq:nc}
N(\cC) = \max_{\bu\neq \bv \in \cC} \abs*{\Ball(\bu) \cap \Ball(\bv)} + 1.
\end{equation}
Trivially, the range of the maximum intersection between error balls of any two distinct codewords in $\cC$ satisfies
\[ 
0 \leq \max_{\bu\neq \bv \in \cC} \abs*{\Ball(\bu) \cap \Ball(\bv)} \leq \max_{\bu\in \cC} \abs*{\Ball(\bu)}.
\]
When the reconstruction degree is $1$, the reconstruction code $\cC$ becomes an error-correcting code.

Throughout the paper, we will often use \emph{the same code} $\cC$ over different channels. The channels will differ by the number of errors they may introduce. In most cases, we will simply write $N(\cC)$ to denote the reconstruction degree of the code when the channel is clear from the context. If we wish to emphasize the number of errors the channel may introduce, $t$, we shall write $N_t(\cC)$.


\begin{definition}
Let $\ch$ be a noisy channel, and $q \geq 2$ be the alphabet size. For a given number of reads, $N\in \N$, a reconstruction code $\cC \subseteq \Sigma_q^n$ is optimal if
\[
    r(\cC) = \min\set{r(\cC'): \cC' \text{ is an $(n,\ch,N)_q$ reconstruction code}}.
\]
The minimum reconstruction redundancy is denoted as $r^*(n,\ch,N)_q$, or $r^*(\ch,N)$ if $n$ and $q$ are understood from context.
\end{definition}

In this paper, we focus on designing reconstruction codes for the channel $\ch(t, b)$ (see Definition~\ref{def:ch}). We take a closer look at the asymptotic behavior of $r^*(n, \ch(t,b), N)_q$ when $q$, $t$, and $b$, are fixed, while $n,N\to\infty$. Not all growth rates of $N$ with respect to $n$ produce interesting results. To study those of interest, we give the following definition:

\begin{definition}
\label{def:ord}
Let $N:\N\to\Z$ be a function. We say that $N(n)$ has $b$-order $s$, denoted by $\Ord_{b}(N(n))=s$, if for all sufficiently large $n$,
\[
\abs*{\Ball_{s,b}(\mathbf{0})} < N(n) = o(\abs*{\Ball_{s+1,b}(\mathbf{0})}).
\]
The constant function $N(n)=1$ will be said to have $b$-order of $-1$, i.e., $\Ord_{b}(1)=-1$.
\end{definition}

As will become clear later, when $\Ord_{b}(N(n))=s$, the minimum reconstruction redundancy, $r^*(n,\ch(t,b),N(n))_q$ will then be shown to scale as $\Theta(\log n)$. Thus, to capture the relevant constant, we define the asymptotic redundancy factor,
\begin{equation}
\label{eq:defrasy}
\ra(\ch(t,b),s)_q = \inf_{N(n):\Ord_{b}(N(n))=s}\liminf_{n\to\infty} \frac{r^*(n,\ch(t,b),N(n))_q}{\log_q n},
\end{equation}
If $q$ is understood from context, we simply write $\ra(\ch(t,b),s)$. We note that $\ra(\ch(t,b),-1)$ denotes the asymptotic redundancy factor when only a single read is allowed, namely, the code is a regular error-correcting code.

\subsection{The list-reconstruction problem}

Assume \( \cC \subseteq \Sigma_q^n \) is a code, the channel is \( \ch(t,b) \), and consider $N$ reads, where $\Ord_{b}(N)=s'$. Denote the reads by  
\[
Y = \set*{\by_1, \ldots, \by_N} \subseteq \Ball_{t,b}(\bx),
\]  
where each read originates from the error ball of some codeword \( \bx \in \cC \). Define \( \bx_1, \ldots, \bx_{p(Y)} \in \cC \) as 
\[
\set*{\bx_1,\dots,\bx_{p(Y)}}=\cC\cap \bigcap_{i=1}^N \Ball_{t,b}(\by_i).
\]
The set \( Y \) is said to be \emph{list-reconstructed} by the list \( \set{\bx_1, \ldots, \bx_{p(Y)}} \), where \( p(Y) \) represents the list size. Clearly,  
\[
\set*{\by_1, \ldots, \by_N} \subseteq \bigcap_{i=1}^{p(Y)} \Ball_{t,b}(\bx_i).
\]  
The list-reconstruction problem seeks to determine the maximum possible list size, \( p(Y) \), over all choices of \( N \) reads. This is denoted as \( L(n, \cC, \ch(t,b), N)_q \), which is formally defined as  
\[
L(n, \cC, \ch(t,b), N)_q = \max_{Y} \set*{ p(Y) : Y \subseteq \Ball_{t,b}(\bx) \text{ for some } \bx \in \cC, \abs*{Y} = N }.
\]  
We put these notations together and say that a code \( \cC \subseteq \Sigma_q^n \) is an \((n, \ch(t,b), N)_q\) list-reconstruction code with a list size of \( L(n, \cC, \ch(t,b), N)_q \).

It is evident that if there exist \( p \) distinct codewords \( \bx_1, \ldots, \bx_p \in \cC \) such that  
\[
\abs*{ \bigcap_{i=1}^p \Ball_{t,b}(\bx_i)} = N,
\]  
then \( L(n, \cC, \ch(t,b), N)_q \geq p \). Moreover, \( L(n, \cC, \ch(t,b), N)_q \) is a non-increasing function with respect to \( N \).

In this paper, we focus on the channel $\ch(t,b)$ and a $(t-s-1,b)$-burst-correcting code $\cC$, where $0 \leq s \leq t-1$. Let $L(n,t,b,s,N)_{q}$ denote the maximum list size among all $(t-s-1,b)$-burst-correcting codes in $\Sigma_q^n$ that are also $(n,\ch(t,b),N)_q$ list-reconstruction codes. Formally, this can be expressed as:  
\[
L(n,t,b,s,N)_{q} = \max_{\cC}\set*{L(n,\cC,\ch(t,b),N)_q : \cC \text{ is a } (t-s-1,b)\text{-burst-correcting code}}.
\]  
For simplicity, we denote it as $L(s,N)$ when $n$, $t$, $b$, and $q$ are known from context. It is evident that $L(s,N)$ is also a non-increasing function with respect to $N$.

Intuitively, $L(s,N)$ gives us a guarantee of a sufficient list size, when using any $(t-s-1,b)$-burst-correcting code, over $\ch(t,b)$, with $N$ reads. In an analogous manner to the definition of $\ra(\ch(t,b),s)$, we would like to find the asymptotic growth rate of $L(s,N)$ as $n,N\to\infty$. We therefore define
\[
\La(h) = \inf_{N(n):\Ord_{b}(N(n))=s-h}\liminf_{n\to\infty} \frac{\log_q L(s,N(n))}{\log_q n}.
\]

\section{\texorpdfstring{$(t,b)$}{}-Burst-Correcting Codes}\label{sec:ecc}

In this section, we study $(t,b)$-burst-correcting codes via combinatorial arguments. We find a sharp asymptotic estimation of the size of an error ball in the channel $\ch(t,b)$, over an alphabet $\Sigma_{q}$, for fixed positive integers $t$, $b$, up to a constant factor, as $n\to\infty$. This results in asymptotic bounds on the redundancy of any $(t,b)$-burst-correcting code. 

We begin by bounding the size of an error ball. When $t=1$, it is folklore that the size of single $b^\leq$-burst error ball is $1 + n(q-1)q^{b-1}$. We turn to consider up to $t$ $b^{\leq}$-bursts. We note that due to trivial translation invariance, the size of a ball $\abs{\Ball_{t,b}(\bx)}$, does not depend on the choice of the center, $\bx$.

\begin{theorem}\label{thm:ball size}
For any fixed positive integers $q,t,b$, channel $\ch(t,b)$, and any $\bx\in\Sigma_q^n$,
\[
(1-o(1))\parenv*{1-\frac{1}{q}}^t(q^b-1)^t\binom{n}{t} \leq \abs*{\Ball_{t,b}(\bx)} \leq (1+o(1))(q^{b} - 1)^t \binom{n}{t},
\]
as $n\to\infty$, and in particular,
\[
\abs*{\Ball_{t,b}(\bx)}=\Theta(n^t).
\]
\end{theorem}

\begin{proof}[Proof of~\cref{thm:ball size}]
By translation invariance, we may assume \( \bx = \boldsymbol{0} \) without loss of generality. By Lemma~\ref{lem:disjoint}, it follows that  
\[
\abs*{\Ball_{t,b}(\boldsymbol{0})} = \abs*{B_{\leq t, \leq b}},
\]  
where  
\[
B_{\leq t, \leq b} = \bigcup_{i=0}^{t} B_{i, \leq b}.
\]  

For each error in \( B_{k, \leq b} \) with \( k \in [0, t] \), we define the collection of \( k \) $b^\leq$-bursts, which start \emph{exactly} at positions \( 1 \leq p_1 < p_2 < \dots < p_k \leq n \) with respective lengths \emph{exactly} \( 1 \leq \ell_1, \ell_2, \dots, \ell_k \leq b \), as
$E_{\set{(p_1, \ell_1), \dots, (p_k, \ell_k)}}$. Here, by ``exactly'', we mean that positions $p_i$ and $p_{i}+{\ell_i}-1$ are non-zero for all $i\in[k]$. Since there are \( k \) disjoint $b^{\leq}$-bursts with cyclic intervals $[p_1, p_1 + \ell_1 - 1],\dots,[p_k,p_k+\ell_k-1]$, respectively, we obtain the ordering constraint  
\[
1 \leq p_1 \leq p_1 + \ell_1 - 1 < p_2 \leq p_2 + \ell_2 - 1 < \dots < p_k \leq p_k + \ell_k - 1 < p_1 + n.
\]  
In particular, we must have \( p_k \leq n \).

Define  
\[
\cI = \set*{ \set*{(p_1, \ell_1), \dots, (p_k, \ell_k)} : 1 \leq p_1 < \dots < p_k \leq n, 1 \leq \ell_1, \dots, \ell_k \leq b }.
\]  
To facilitate counting, let \( \cJ \) be the subset of \( \cI \) given by  
\[
\cJ = \set*{ \set*{(p_1, \ell_1), \dots, (p_k, \ell_k)} \in \cI : p_{i+1} \geq p_i + 2b - 1 \text{ for } i \in [k], \text{ where } p_{k+1} = p_1 + n }.
\]  
The following observation follows directly from these definitions.
\begin{claim}
  $  \bigcup_{I\in \cJ} E_I \subseteq B_{k,\leq b}  \subseteq \bigcup_{I\in \cI} E_I$.
\end{claim}
\begin{poc}
By definition, it is clear that $B_{k,\leq b}  \subseteq \bigcup_{I\in \cI} E_I$. For the other direction, $I = \set{(p_1,\ell_1),\ldots, (p_k,\ell_k)}\in \cJ$ must satisfy 
  $p_{i+1} \geq p_i + 2b-1 > p_i + \ell_i - 1 \geq p_i$ for any $i\in [k]$, which yields that $E_I \subseteq B_{k,\leq b}$.
\end{poc}

Moreover, we have the following crucial observation.
\begin{claim}\label{Claim:Disjoint}
    For any $I,I' \in \cJ$, $E_I \cap E_{I'} \neq \emptyset$ if and only if $I = I'$.
\end{claim}

\begin{poc}
Let \( I = \set{ (p_1, \ell_1), \ldots, (p_k, \ell_k) }\in\cJ \) and \( I'= \set{ (p_1^\prime, \ell_1^\prime), \ldots, (p_k^\prime, \ell_k^\prime) }  \in \cJ \), and assume $\be\in E_I\cap E_{I'}$ is an error pattern. Define the multiset  
\[
Q = \mset{ p_1, p_1 + \ell_1 -1, \ldots, p_k, p_k + \ell_k -1 },
\]  
where \( \mset{\cdot} \) denotes a multiset. Note that \( Q \) reduces to a set if \( \ell_i > 1 \) for all \( i \in [k] \). Note that since $\be\in E_I$, the positions of $\be$ that are in $Q$ contain non-zero entries. However, these positions must fall within the $b^{\leq}$-bursts corresponding to $I'$, a restriction which we now proceed to examine.

For each \( i \in [k] \), define  
\[
R_i = Q \cap [p_i^\prime, p_i^\prime + \ell_i^\prime -1],
\]  
which represents the sub-multiset of \( Q \) contained within the interval \( [p_i^\prime, p_i^\prime + \ell_i^\prime -1] \). We claim that \( R_i \) contains at most two distinct elements for any \( i \in [k] \). Otherwise, there exists some \( i \in [k] \) such that \( R_i \) contains at least three elements, appearing in one of the following two possible forms:  
\begin{align*}
    &\mset{ p_j, p_j+\ell_j-1, p_{j+1} },\\
    &\mset{ p_j+\ell_j-1, p_{j+1}, p_{j+1}+\ell_{j+1}-1 },
\end{align*}  
where \( 1 \leq j \leq k \). 

In both cases, we obtain the bound  
\[
\ell_i^\prime \geq \min \set*{ p_{j+1} - p_j + 1, p_{j+1}+\ell_{j+1}-p_{j}-\ell_{j} + 1 }.
\]  
Since \( p_{j+1} - p_j \geq 2b-1 \) and \( 1 \leq \ell_j, \ell_{j+1} \leq b \), it follows that  
\[
\ell_i^\prime \geq b+1,
\]  
which contradicts the assumption that \( \ell_i^\prime \leq b \).
  
 On the other hand, the average size of \( R_i \) is  
\[
\frac{\abs*{Q}}{k} = 2,
\]  
which implies that \( \abs{R_i} = 2 \) for all \( 1 \leq i \leq k \). Consequently, \( R_i \) must belong to the set  
\[
\set*{ \mset{p_1, p_1 + \ell_1 - 1}, \mset{p_1 + \ell_1 - 1, p_2}, \mset{p_2, p_2 + \ell_2 - 1}, \dots, \mset{p_k, p_k + \ell_k - 1}, \mset{p_k + \ell_k - 1, p_1} }.
\]  

The gap between \( p_j + \ell_j - 1 \) and \( p_{j+1} \) is given by  
\[
p_{j+1} - p_j - \ell_j + 2 \geq 2b - 1 - b + 2 = b+1.
\]  
Thus, \( R_i \) must take the form \( \mset{p_j, p_j + \ell_j - 1} \) for some \( j \in [k] \), which implies  
\[
[p_j, p_j + \ell_j - 1] \subseteq [p_i^\prime, p_i^\prime + \ell_i^\prime - 1].
\]  
By the increasing order of \( p_1, \dots, p_k \) and \( p_1^\prime, \dots, p_k^\prime \), we deduce that  
\[
[p_i, p_i + \ell_i - 1] \subseteq [p_i^\prime, p_i^\prime + \ell_i^\prime - 1] \quad \text{for } 1 \leq i \leq k.
\]  
Applying the same argument, we also obtain  
\[
[p_i^\prime, p_i^\prime + \ell_i^\prime - 1] \subseteq [p_i, p_i + \ell_i - 1] \quad \text{for } 1 \leq i \leq k,
\]  
which implies \( I = I' \). This completes the proof. 
\end{poc}

By Claim~\ref{Claim:Disjoint}, we have 
\[
  \abs*{\bigcup_{I\in \cJ} E_I} = \sum_{I\in \cJ} \abs*{E_I}.
\]
For each $I =\set{ (p_1,\ell_1),\ldots, (p_k,\ell_k)} \in \cI$, by definition of $\cI$ we have 
  \[
  \abs{E_I} \leq \prod_{i=1}^k (q-1)q^{\ell_i - 1} = (q-1)^k q^{-k + \sum_{i=1}^k \ell_i}.
  \]
Therefore, we have
   \begin{align}
      \abs*{B_{k,\leq b}} &\leq \sum_{I\in \cI} \abs*{E_I} \leq \sum_{\substack{ 1 \leq p_1 <\cdots < p_k \leq n \\ 1\leq \ell_1, \ldots,\ell_k \leq b}} (q-1)^k q^{-k + \sum_{i=1}^k \ell_i} \nonumber \\
      & = \binom{n}{k} \cdot (q-1)^k q^{-k} \sum_{1\leq \ell_1, \ldots,\ell_k \leq b} q^{\sum_{i=1}^k \ell_i} \nonumber \\
      &= \binom{n}{k}  (q^{b} - 1)^k, \label{eq:claimupper}
  \end{align}
 where in the last equality, we take advantage of the fact that $\sum_{1\leq \ell_1, \ldots,\ell_k \leq b} q^{\sum_{i=1}^k \ell_i}$ is the expansion of $(q + q^2 + \cdots + q^b)^k = \frac{(q^{b+1} - q)^k}{(q-1)^k}$.
  
We now turn to consider a lower bound on $\abs{B_{k,\leq b}}$. For each $I\in \cJ$, by definition of $\cJ$, we have 
  \[
  \abs{E_I} \geq \prod_{i=1}^k (q-1)^2 q^{\ell_i - 2} = (q-1)^{2k} q^{-2k + \sum_{i=1}^k \ell_i}.
  \]
Thus,
    \begin{align*}
      \abs*{B_{k,\leq b}} &\geq \sum_{I\in \cJ} \abs*{E_I} \geq \sum_{ \substack{ p_{i} +2b-1 \leq p_{i+1} \\ 1\leq i \leq k}} \sum_{1\leq \ell_1, \ldots,\ell_k \leq b} (q-1)^{2k} q^{-2k + \sum_{i=1}^k \ell_i} \\
      & = \sum_{ \substack{ p_{i} +2b-1 \leq p_{i+1} \\ 1\leq i \leq k}} (q-1)^{2k} q^{-2k} \sum_{1\leq \ell_1, \ldots,\ell_k \leq b} q^{\sum_{i=1}^k \ell_i} \\
      &= \sum_{ \substack{ p_{i} +2b-1 \leq p_{i+1} \\ 1\leq i \leq k}} \frac{(q-1)^k (q^{b+1} - q)^k}{q^{2k}}. \\
      &\geq \binom{n-(2b-2)k}{k}\cdot  \frac{(q-1)^k (q^{b} - 1)^k}{q^{k}},
  \end{align*}
 where for the final inequality, we under-count by marking $k$ positions, immediately inserting $2b-2$  positions after each, and then insert the remaining $n-k(2b-1)$, which is equivalent to throwing $n-k(2b-1)$ balls into $k+1$ bins.

Rewriting the lower bound, and combining it with the upper bound of~\eqref{eq:claimupper} we obtain,
\[
\parenv*{1-\frac{1}{q}}^k(q^b-1)^k\binom{n}{k}(1-o(1)) \leq\abs*{B_{k,\leq b}} = (q^{b} - 1)^k \binom{n}{k},
\]
where we assume $k,b,t,q$ are constants and $n\to\infty$. Summing over all $k\in[0,t]$ for the upper bound, we obtain 
\begin{align*}
  \parenv*{1-\frac{1}{q}}^t(q^b-1)^t\binom{n}{t}(1-o(1)) &= \abs*{B_{t,\leq b}} \leq \abs*{B_{\leq t,\leq b}} \leq  \sum_{k=0}^t \abs{B_{k,\leq b}} = \sum_{k=0}^t (q^{b} - 1)^k \binom{n}{k} \\
  & = (1+o(1))(q^b-1)^t\binom{n}{t}.
\end{align*}
\end{proof}

\begin{rmk}
\label{rmk:ballsizecomb}
If we wish for a non-asymptotic version of Theorem~\ref{thm:ball size}, by examining its proof we can see that for all $n\geq 2bt$,
\[
\parenv*{1-\frac{1}{q}}^t(q^b-1)^t\binom{n-(2b-2)t}{t} \leq \abs*{B_{\leq t,\leq b}} \leq  (t+1)(q^{b} - 1)^t \binom{n}{t}.
\]
\end{rmk}

With Theorem~\ref{thm:ball size} established, one can utilize the classical argument underpinning the sphere-packing bound and the Gilbert–Varshamov (GV) bound to derive immediate bounds on the redundancy of $(t,b)$-burst-correcting codes.

\begin{cor}\label{cor:Redu}
For any fixed positive integers $q,t,b$, and channel $\ch(t,b)$, the minimum redundancy of $(t,b)$-burst-correcting codes of length $n$ over $\Sigma_q$ is at least $t\log_q n + \Omega(1)$ and at most $2t\log_q n + O(1)$, namely,
\[
t\leq \ra(\ch(t,b),-1)\leq 2t.
\]
\end{cor}

\begin{proof}[Proof of~\cref{cor:Redu}]
By standard arguments for the sphere-packing bound and GV bound~\cite{macwilliams1977theory}, the size of any code does not exceed $q^n/\abs{\Ball_{t,b}(\bx)}$, and there exists a code of size at least $q^n/\abs{\Ball_{2t,b}(\bx)}$. We then simply apply the bound of Theorem~\ref{thm:ball size}.
\end{proof}

As a final note for this section, it is well known~\cite{abdel1986existence,abdel1988existence,etzion2001constructions} that there is not difference in the asymptotic redundancy factor between a single burst and a single substitution channel,
\[
\ra(\ch(1,b),-1)=\ra(\ch(1),-1).
\]
We conjecture this is also the case for multiple substitutions/bursts.

\begin{conjecture}\label{conj: burst}
    Let $ q, t, b$ be constant positive integers with $ q \geq 2 $. Then,
    \[
    \ra(\ch(t,b),-1)=\ra(\ch(t),-1).
    \]
\end{conjecture}

\section{Reconstruction Codes in \texorpdfstring{$\ch(t,b)$}{}}\label{sec:reconstruction code}

In this section, we examine the trade-off between the redundancy of a reconstruction code and its reconstruction degree. Specifically, we aim to determine $\ra(\ch(t,b),s)$. For the lower bound, we carefully analyze the difference between codewords. For the upper bound, we build a bridge between burst-reconstruction codes and burst-correcting codes.

We first observe that for every \( s \geq t \), we immediately have
\[
\ra(\ch(t,b), s) = 0,
\]
since we can take a code $\cC=\Sigma_q^n$. Thus, in the following discussion, we may assume \( 0 \leq s \leq t - 1 \).  We also recall that by Theorem~\ref{thm:ball size}, $\abs{\Ball_{t,b}(\bx)}=\Theta(n^t)$. Thus, a function $N(n)$ such that $\Ord_{b}(N(n))=s$ is equivalent to having $N(n)=\Omega(n^s)$, $N(n)=o(n^{s+1})$, and that for all sufficiently large $n$ we have $N(n)>\abs{\Ball_{s,b}(\bx)}$.

We first describe a simple trade-off between reconstruction degree and redundancy, in substitution channels, i.e., $b=1$.

\begin{theorem}\label{thm:MARR for sub}
    Given positive integers $q,t$ with $q \geq 2$, and $\ch(t)$ (the channel with at most $t$ substitutions), we have 
    \[
    \ra(\ch(t),s) =  \ra(\ch(t-s-1),-1),
    \]
    for any integer $-1\leq s \leq t-1$. 
\end{theorem}
\begin{proof}[Proof of~\cref{thm:MARR for sub}]
In one direction, if $\Ord_{1}(N(n))=-1$, by definition, $N(n)=1$ for all $n$, and thus the $\inf$ part of~\eqref{eq:defrasy} is degenerate. Consider a sequence of error-correcting codes $\cC_n$ of length $n$, over $\ch(t-s-1)$, such that
\[
r(\cC_n)=r^*(n,\ch(t-s-1,1),1)_q,
\]
and therefore,
\[
\liminf_{n\to\infty}\frac{r(\cC_n)}{\log_q n}=\ra(\ch(t-s-1,-1).
\]
Since these are error-correcting codes, $N_{t-s-1}(\cC_n)=1$.

By Levenshtein~\cite{levenshtein2001efficient}, if $\bu,\bv\in\Sigma_q^n$ and $d_1(\bu,\bv)=d$, then the size of the intersection of the balls of radius $d$ around $\bu$ and $\bv$ does not depend on the choice of $\bu$ and $\bv$, and is equal to
\begin{align}
\cI_t(n,d) &= \abs*{\Ball_{t,1}(\bu) \cap \Ball_{t,1}(\bv)} \nonumber \\
&= \sum_{i=0}^{t - \ceil{\frac{d}{2}}} \binom{n-d}{i}(q-1)^i \cdot \parenv*{\sum_{k = d-t+i}^{t-i} \sum_{\ell = d-t+i}^{t-i} \binom{d}{k} \binom{d-k}{\ell} (q-2)^{d-k-\ell}} \nonumber \\
&= \Theta( n^{t - \ceil{d/2}}). \label{eq:Ind}
\end{align}
Let us now use the codes $\cC_n$ over the channel $\ch(t,1)$. In this setting, with at most $t$ errors present, for any $\cC \subseteq \Sigma_q^n$ with minimum Hamming distance $d$, we have
\[
N_t(\cC) = 1+\cI_t(n,d) = \Theta(n^{t-\ceil{d/2}}).
\]
Since our codes $\cC_n$ have minimum distance at least $2(t-s)-1$ (to be able to correct $t-s-1$ errors by the channel), we have
\[N_t(\cC_n)\leq 1+\cI_t(n,2(t-s)-1)=O(n^s).\]
We now define the sequence
\[N'(n)=\max\set*{N_t(\cC_n),\abs*{\Ball_{s,1}(\mathbf{0})}+1}.\]
By the above discussion and Theorem~\ref{thm:ball size}, we have $N'(n)=O(n^s)$, as well as $\Ord_{1}(N'(n))=s$. Now, by definition
\[
\ra(\ch(t),s) \leq \liminf_{n\to\infty} \frac{r(\cC_n)}{\log_q n}= \ra(\ch(t-s-1),-1).
\]

In the other direction, consider a sequence of reconstruction codes, $\cC_n$, of length $n$, and a sequence $N(n)$ such that $\Ord_1(N(n))=s$ and $\cC_n$ can uniquely reconstruct sequences transmitted over $\ch(t,1)$ with $N(n)$ reads. By Theorem~\ref{thm:ball size}, we have $N(n)=o(n^{s+1})$.

Assume to the contrary, that for infinitely many values of $n$, $\cC_n$ contains two distinct codewords $\bu,\bv\in\cC_n$ with $d_1(\bu,\bv)\leq 2(t-s-1)$. However, by~\eqref{eq:Ind}, for this subsequence of codes, $N_t(\cC_n)=\Omega(n^{s+1})$, contradicting $N(n)=o(n^{s+1})$. Thus, there exists $n_0\in\N$ such that for all $n\geq n_0$, the minimum distance of $\cC_n$ is at least $2(t-s)-1$. For convenience we remove the codes of length $n<n_0$ from the sequence.

If we now use $\cC_n$ over $\ch(t-s-1)$, then $N_{t-s-1}(\cC_n)=1$. Thus, by definition,
\[
\ra(\ch(t-s-1),-1)\leq \liminf_{n\to\infty} \frac{r(\cC_n)}{\log_q n}.
\]
Since this holds for any sequence of codes $\cC_n$, and any $N(n)$ with the requirements set above, we have
\[
\ra(\ch(t-s-1),-1)\leq \inf_{N(n):\Ord_{1}(N(n))=s}\liminf_{n\to\infty} \frac{r^*(n,\ch(t),N(n))_q}{\log_q n}=\ra(\ch(t),s).
\]
\end{proof}

We derive an upper bound on $\ra(\ch(t),s)$ by known constructions from the theory of error-correcting codes. More precisely, when
\begin{align*}
&n = q^m-1, \text{ for } m\in \N,\\
&n-k \leq \begin{cases}
    \ceil{\frac{d-1}{2}}m, &q = 2,\\
    2\ceil{\frac{d-1}{2}}m, &q > 2,
\end{cases}
\end{align*}
by taking advantage of the classic primitive narrow-sense $[n,k,d]_q$ BCH codes~\cite{macwilliams1977theory}, we have  
\[
\ra(\ch(t-s-1),-1) \leq \begin{cases}
    (t-s-1), &q = 2,\\
    2(t-s-1), &q > 2.
\end{cases}
\]
On the other hand, we can easily obtain the lower bound by applying the classical sphere-packing bound, which states that any $q$-ary $n$-length code with minimum Hamming distance $d$ has redundancy at least $\floor{\frac{d-1}{2}} \log_q n + O(1)$. We thus have the following:

\begin{theorem}
Let $t,q$ be positive integers with $q \geq 2$, and let $s$ be an integer with $-1\leq s \leq t-1$. Then
    \[
    t-s-1 \leq \ra(\ch(t),s)_q \leq 2(t-s-1).
    \]
    Moreover for $q=2$,
    \[
    \ra(\ch(t),s)_2 = t-s-1.
    \]
\end{theorem}

\subsection{A lower bound on  \texorpdfstring{$\ra(\ch(t,b),s)$}{}}

To derive a lower bound on $\ra(\ch(t,b),s)$, we investigate the types of structures presented in reconstruction codes, that prevent successful codeword reconstruction within $o(n^{s+1})$ reads. 

\begin{lemma}\label{lem:forbidden}
Let $q,b,t$ be constant positive integers, $q \geq 2$. Let $\cC_n \subseteq \Sigma_q^n$ be a sequence of reconstruction codes over $\ch(t,b)$. If $N_t(\cC_n)=o(n^{s+1})$ then for all sufficiently large $n$, $\cC_n$ is also a  $(t-s-1,b)$-burst-correcting code.
\end{lemma}
\begin{proof}[Proof~\cref{lem:forbidden}]
Let us define
\[
\cC_n^\Delta = \set*{\bx - \by: \bx,\by \in \cC_n}.
\]
We first contend that if $N_t(\cC_n)=o(n^{s+1})$, then $\cC_n^\Delta \cap B_{\leq 2(t-s-1), \leq b}=\emptyset$. 

Assume to the contrary that there exists some \( \bw \in \cC_n^\Delta \cap B_{\leq 2(t-s-1), \leq b} \). Then, we have \( \bw = \bx - \by \) for some two codewords \( \bx, \by \in \cC_n \). Since \( \bw \in B_{\leq 2(t-s-1), \leq b} \), we can also express it as \( \bw = \bu - \bv \), where  
\[
\bu = \be_1 + \cdots + \be_i, \quad \bv = \be_1^\prime + \cdots + \be_j^\prime,
\]  
with \( i, j \leq t-s-1 \), and the set \( \set{ \be_1, \dots, \be_i, \be_1^\prime, \dots, \be_j^\prime } \) consists of pairwise disjoint $b^\leq$-bursts. 

The elements of $[n]$ \emph{outside} of these $i+j$ intervals, may be partitioned into at most $i+j$ pair-wise disjoint cyclic intervals. By simple averaging, at least one such cyclic interval, denoted as \( I = [f, g] \) (assuming without loss of generality that \( f < g \)), has a length of at least  
\[
\frac{n - (i+j)b}{i+j} \geq \frac{n}{2(t-s-1)} - b.
\]

Taking this interval $I$, let us now attempt to place $s+1$ more $b^{\leq}$-bursts inside of it. Consider arbitrary $s+1$ disjoint $b^{\leq}$-bursts, $\bz_1,\ldots,\bz_{s+1}$, contained within $I$. Let $I_k = I(\bz_k)$ for $k\in[s+1]$. By considering merely the placement of these bursts (and ignoring the error values), the number of ways to choose $\bz_1,\ldots,\bz_{s+1}$ is at least the number of solutions $w_1,\ldots,w_{s+1}$ to
\[
    f \leq w_1 < w_1 +b \leq w_2 < w_2+b \leq \ldots \leq w_{s+1} < w_{s+1}+b \leq g.
\]
This number is exactly $\binom{\abs{I} - (b-1)(s+1)}{s+1} = \Omega(\abs{I}^{s+1})=\Omega(n^{s+1})$.
    
Let $S$ consist of all possible pairs $(\bu',\bv')$, where 
\begin{align*}
    &\bu' = \be_1 + \cdots + \be_i + \bz_1 + \cdots + \bz_{s+1} \in B_{i+s+1,\leq b} \subseteq B_{\leq t,\leq b}\\
    &\bv' = \be_1^\prime + \cdots + \be_j^\prime + \bz_1 + \cdots + \bz_{s+1} \in B_{j+s+1,\leq b} \subseteq B_{\leq t,\leq b}.
    \end{align*}
   Note that for any pair $(\bu',\bv')$, we have $\bx - \by = \bw = \bu' - \bv'$. Therefore 
    \[
    \abs*{\Ball_{t,b}(\bx) \cap \Ball_{t,b}(\by) } \geq \abs{S} = \Omega(n^{s+1}),
    \]
which contradicts the reconstruction degree $N_t(\cC_n)=o(n^{s+1})$.

We therefore have for all sufficiently large $n$, $\cC_n^\Delta \cap B_{\leq 2(t-s-1), \leq b}=\emptyset$. Now, assume to the contrary that for these $n$, $\cC_n$ is not a $(t-s-1,b)$-burst-correcting code. Then there exist $\bx,\by\in\cC_n$, and two errors pattern $\bu,\bv\in B_{\leq t-s-1,\leq b}$, such that $\bx+\bu=\by+\bv$. But that implies $\bx-\by=\bv-\bu$. Note now that $\bx-\by\in\cC_n^\Delta$, and $\bv-\bu\in B_{\leq 2(t-s-1),\leq b}$, contradicting the fact that $\cC_n^\Delta \cap B_{\leq 2(t-s-1), \leq b}=\emptyset$.
\end{proof}

Therefore, we have the following lower bound.
\begin{theorem}\label{thm:lower}
Let $q,b,t$ be constant positive integers, $q \geq 2$. Then, for all $0\leq s\leq t-1$,
\[
\ra(\ch(t,b),s) \geq  \ra(\ch(t-s-1,b),-1)\geq t-s-1.
\]
\end{theorem}
\begin{proof}[Proof of~\cref{thm:lower}]
Simply combine Lemma~\ref{lem:forbidden} with Corollary~\ref{cor:Redu}, using the same logic as the proof of Theorem~\ref{thm:MARR for sub}.
\end{proof}

\subsection{An upper bound on  \texorpdfstring{$\ra(\ch(t,b),s)$}{}}

To construct a reconstruction code $\cC \subseteq \Sigma_q^n$ with a reconstruction degree of $o(n^{s+1})$, it is necessary for $\cC$ to be a $(t-s-1,b)$-burst-error correcting code. In this subsection, we establish that this requirement is both necessary and sufficient. In other words, any $(t-s-1,b)$-burst-error correcting code $\cC \subseteq \Sigma_q^n$ inherently possesses a reconstruction degree of $O(n^s)$.

\begin{theorem}\label{thm:construction degree}
Let $q,b,t$ be constant positive integers, with $q \geq 2$. If $\cC \subseteq \Sigma_q^n$ is a $(t-s-1,b)$-burst-correcting code, $0\leq s \leq t-1$, then the reconstruction degree of $\cC$ satisfies
\[
     N(\cC) \leq (t+1)^2 f_2 n^s = O(n^s),
\]
where $f_2 = \frac{t^{t-s} 2^{t-s+1} b^{2(t-s)} q^{bt}}{s!}$.
\end{theorem}

\begin{proof}[Proof of~\cref{thm:construction degree}]
Recall~\eqref{eq:nc},
\[
N(\cC) = \max_{\bx\neq \by \in \cC} \abs*{\Ball(\bx) \cap \Ball(\by)} + 1.
\]  
Thus, it suffices to establish an upper bound on the size of the intersection of two error balls centered at distinct codewords, $\bx,\by\in\cC$. Define
\[
M_{\leq t} = \set*{ (\balpha, \bbeta) : \balpha, \bbeta \in B_{\leq t, \leq b}, \, \bx - \by = \bbeta - \balpha }.
\]  
Since
\[
\Ball_{t,b}(\bx) = \set*{ \bx + \be : \be \in B_{\leq t, \leq b} }, \quad \Ball_{t,b}(\by) = \set*{ \by + \be : \be \in B_{\leq t, \leq b} },
\]  
we have
\[
\abs*{ \Ball_{t,b}(\bx) \cap \Ball_{t,b}(\by) } = \abs*{M_{\leq t}}.
\]

For each \( 0 \leq i, j \leq t \), define  
\[
M_{i,j} = \set*{ (\balpha, \bbeta) : \balpha \in B_{i, \leq b}, \, \bbeta \in B_{j, \leq b}, \, \bx - \by = \bbeta - \balpha }.
\]  
Then, we obtain the bound  
\[
\abs*{M_{\leq t}} \leq \sum_{0 \leq i, j \leq t} \abs*{M_{i,j}}.
\]
Since $\cC$ is a $(t-s-1,b)$-burst-correcting code, necessarily $M_{i,j} = \emptyset$ for any $i+j \leq 2(t-s-1)$. Otherwise, there would exist a pattern of $i$ $b^{\leq}$-bursts $\balpha\in B_{i,\leq b}$, and a pattern of $j$ $b^{\leq}$-bursts $\bbeta\in B_{j,\leq b}$, such that $\bx-\by=\bbeta-\balpha$, implying that $d_b(\bx,\by)\leq 2(t-s-1)$, and contradicting the minimum distance of the code.

For $i \leq s$, we have 
\[
\abs*{M_{i,j}} \leq \abs*{\set*{\balpha : \balpha \in B_{i,\leq b}}} \leq \abs*{B_{s,\leq b}} \leq (1+o(1))  (q^{b} - 1)^s \binom{n}{s}\leq (1+o(1)) f_1 n^s,
\]
where $f_1 = \frac{(q^{b} - 1)^s }{s!}$, and we take advantage of Theorem~\ref{thm:ball size} in the third inequality. Applying the same argument, we obtain the same bound on $\abs{M_{i,j}}$ for any $j \leq s$.

We now consider the remaining cases of $i,j\geq s+1$. At this point we would like to introduce the notion of a \emph{sufficient burst pattern} (w.r.t. $\bx-\by$). Let $\bgamma\in B_{\ell,\leq b}$ be a burst pattern. We say $\bgamma$ is $\ell$-sufficient for given integer $s$, if there exist pair-wise disjoint $b^{\leq}$-bursts, $\bc_1,\dots,\bc_\ell$, $\bgamma=\bc_1+\dots+\bc_\ell$, such that $\supp(\bc_k)\cap\supp(\bx-\by)\neq\emptyset$, for all $k\in[\ell-s]$. Namely, if we can find $\ell-s$ positions $p_k\in\supp{\bc_k}$, for all $k\in[\ell-s]$, such that $p_k\in\supp(\bx-\by)$. We call $p_1,\dots,p_{\ell-s}$ \emph{sufficient representatives of $\bgamma$} (w.r.t. to the partition $\bc_1,\dots,\bc_\ell$).

For any $(\balpha,\bbeta) \in M_{i,j}$, by definition, we can write $\balpha$ and $\bbeta$ as,
\begin{align*}
     \balpha &= \ba_1 + \cdots + \ba_i, & 
     \bbeta  &= \bb_1 + \cdots + \bb_j,
\end{align*}
where $\ba_1,\dots,\ba_i$ are pairwise disjoint $b^{\leq}$-bursts, and similarly for $\bb_1,\dots,\bb_j$.
     
Let $m_1$ be the number of $b^\leq$-bursts from $\set{\ba_1,\ldots, \ba_i}$ whose support is contained in $\supp(\bbeta)$, and let $m_2$ be the number of $b^\leq$-bursts from $\set{\bb_1,\ldots, \bb_j}$ whose support is contained in $\supp(\balpha)$. Without loss of generality, we assume 
\begin{align*}
         {\supp(\ba_{1}),\ldots,\supp(\ba_{m_1})} &\subseteq \supp(\bbeta),\\
         {\supp(\bb_{1}),\ldots,\supp(\bb_{m_2})} &\subseteq \supp(\balpha).
\end{align*}  

\begin{claim}\label{claim}
$\balpha$ is $i$-sufficient, or $\bbeta$ is $j$-sufficient.
\end{claim}

\begin{poc}
We divide the proof of claim according to the values of $m_1$ and $m_2$.

\textbf{Case 1:} If $m_1 \leq s$, then $\balpha$ contains $i-m_1 \geq i-s$ disjoint $b^\leq$-bursts, $\ba_{m_1+1}, \ldots,\ba_{i}$, none of whose supports is contained in $\supp(\bbeta)$. Therefore we can choose $i-m_1$ points $p_k \in \supp(\ba_{m_1+k}) \setminus \supp(\bbeta) \subseteq \supp(\bbeta - \balpha) = \supp(\bx-\by)$ for $k\in [i-m_1]$. 

\textbf{Case 2:} If $m_2 \leq s$, we can apply the identical argument as above.

\textbf{Case 3:} If $m_1,m_2 \geq s+1$, let us now denote by $r$ the number of $b^\leq$-bursts from $\set{\ba_1,\ldots, \ba_i}$ whose supports are disjoint from $\supp(\bx - \by)$. Without loss of generality, assume $\supp(\ba_\ell)\cap \supp(\bx-\by) = \emptyset$ for all $\ell \in [r]$. We claim that $r \leq s$.

Assume to the contrary, that $r\geq s+1$. Since the supports of $\bb_1,\ldots,\bb_{m_2}$ are contained in $\balpha$, we have 
    \[
    \supp\parenv*{\sum_{k=1}^{m_2}\bb_k - \balpha} \subseteq \supp(\balpha).
    \]
Define $\supp^*(\bb_{\ell}) = \supp(\bb_{\ell}) \setminus \supp(\balpha)$ for $\ell\in [m_2+1,j]$, obviously we have 
    \[
    \bigcup_{\ell = m_2+1}^j \supp^*(\bb_\ell) \cap \supp(\balpha) = \emptyset.
    \]
     Moreover, we can see that
     \[
    \supp\parenv{\bx - \by} = \supp\parenv{\bbeta - \balpha} \subseteq  \bigcup_{\ell = m_2+1}^j \supp(\bb_\ell) \cup \supp(\balpha) =  \bigcup_{\ell = m_2+1}^j \supp^*(\bb_\ell) \cup \bigcup_{k=1}^i \supp(\ba_k).
    \]
By our assumption, $ \bigcup_{k=1}^r \supp(\ba_k) \cap \supp(\bx-\by) = \emptyset$, which implies that
     \begin{align*}
    \supp\parenv{\bx - \by}  &\subseteq  \bigcup_{\ell = m_2+1}^j \supp^*(\bb_\ell) \cup \bigcup_{k=1}^i \supp(\ba_k) \setminus \bigcup_{k=1}^r \supp(\ba_k) \\
    &= \bigcup_{\ell = m_2+1}^j \supp^*(\bb_\ell) \cup \bigcup_{k=r+1}^i \supp(\ba_k) \subseteq B_{\leq i-r+j-m_2, \leq b}.
    \end{align*}
But now, since $i-r+j-m_2 \leq t - (s+1) + t - (s+1) = 2(t-s-1)$, we have $\supp\parenv{\bx - \by} \in B_{\leq 2(t-s-1), \leq b}$, which is a contradiction to $\cC$ being a $(t-s-1,b)$-burst-correcting code.

It follows that indeed $r\leq s$. Hence, $\supp(\bx-\by)$ contains $i-r \geq i-s$ number of positions $p_1,\ldots,p_{i-r}$, which belong to $\supp(\ba_{r+1}), \ldots,\supp(\ba_i)$, respectively.
\end{poc}

We continue the proof of the main claim. Define the set of sufficient burst patterns as
\[
S_\ell = \set*{ \bgamma\in B_{\ell,\leq b}: \text{$\bgamma$ is $\ell$-sufficient}}.
\]
For any \( (\balpha, \bbeta) \in M_{i,j} \), Claim~\ref{claim} implies that at least one of the conditions \( \balpha \in S_i \) or \( \bbeta \in S_j  \) must hold. Moreover, any two distinct pairs \( (\balpha, \bbeta), (\balpha^\prime, \bbeta^\prime) \in M_{i,j} \) satisfy \( \balpha \neq \balpha^\prime \) and \( \bbeta \neq \bbeta^\prime \). Consequently,  
\[
\abs*{M_{i,j}} \leq \abs*{S_i}+\abs*{S_j}.
\]  

We now upper bound $\abs{S_i}$. A similar argument will bound $\abs{S_j}$. Since \( \abs{\supp(\bx-\by)} = \abs{\supp(\bbeta - \balpha)} \leq 2tb \), it follows that  
\[
\abs*{\binom{\supp(\bx-\by)}{i-s} } \leq (2tb)^{i-s}.
\]  
We construct an auxiliary bipartite graph \( G = S_i \times \binom{\supp(\bx-\by)}{i-s} \), where the vertex set is given by  
\[
V(G) = S_i \cup \binom{\supp(\bx-\by)}{i-s}.
\]  
A pair of vertices \( \balpha \in S_i \) and \( \set{ p_1, \dots, p_{i-s} } \in \binom{\supp(\bx-\by)}{i-s} \) are adjacent if and only if \( p_1, \dots, p_{i-s} \) are sufficient representatives of $\balpha$ (w.r.t. to the partition $\ba_1,\dots,\ba_i$).

We then double count the number of edges in $G$. On one hand, we can see every $\balpha\in S_i$ is adjacent to at least one of the elements in $\binom{\supp(\bx-\by)}{i-s}$, which yields $\abs{E(G)} \geq \abs{S_i}$. 

On the other hand, for any \( \set{ p_1, \ldots, p_{i-s} } \in \binom{\supp(\bx-\by)}{i-s} \), let \( n_{\set{p_1, \ldots, p_{i-s}}} \) denote the number of elements \( \balpha \in S_i \) such that \( \balpha \) is adjacent to \( \set{ p_1, \ldots, p_{i-s} } \). By definition, each \( \balpha = \ba_1 + \cdots + \ba_i \in B_{i,\leq b} \) has \( i-s \) disjoint \( b^\leq \)-bursts \( \set{ \ba_1, \ldots, \ba_{i-s} } \), where each \( p_k \in \supp(\ba_k) \) for \( k \in [i-s] \). Continuing our bounding, given $p_k\in\supp(\ba_k)$, the number of ways completing $\ba_k$ is crudely upper bounded by $b$ choices of the relative position of $p_k$ in the $\ba_k$, and then no more than $q^b$ ways of filling the values of the burst, i.e., no more than $bq^b$ ways. Continuing in this manner for all $k\in[i-s]$, $\ba_1,\dots,\ba_{i-s}$ can be chosen from no more than $(bq^b)^{i-s}$ options. Moreover, by Theorem~\ref{thm:ball size}, the number of ways to choose the remaining \( s \) disjoint \( b^\leq \)-bursts, \( \set{ \be_{i-s+1}, \ldots, \be_i } \), is at most $\abs{B_{s,\leq b}} \leq \binom{n}{s} (q^b - 1)^s$. Thus, we obtain  
\[
n_{\set{p_1, \ldots, p_{i-s}}} \leq \binom{n}{s} (b q^b)^{i-s} (q^b - 1)^s \leq g_1 n^s,
\]  
where $g_1 = \frac{q^{bi} b^{i}}{s! b^{s}}$.

In total, the number of edges in $G$ can be upper bounded as follows:
\[
\abs*{S_i}\leq \abs*{E(G)} = \sum_{\set{p_1,\ldots,p_{i-s}} \in \binom{\supp(\bx-\by)}{i-s}} n_{\set{p_1,\ldots,p_{i-s}}} \leq \sum_{\set{p_1,\ldots,p_{i-s}} \in \binom{\supp(\bx-\by)}{i-s}} g_1 n^s \leq g_1 (2tb)^{i-s} n^s.
\]
A similar argument gives us
\[
\abs*{S_j}\leq  g_2 (2tb)^{j-s} n^s,
\]
where $g_2=\frac{q^{bj}b^j}{s!b^s}$. Then, we have
\[
\abs*{M_{i,j}} \leq \abs*{S_i} + \abs*{S_j} \leq g_1 (2tb)^{i-s} n^s + g_2 (2tb)^{j-s} n^s = \frac{(2tb)^{i} b^{i} q^{bi} + (2tb)^{j} b^{j} q^{bj}}{s!(2tb)^s b^{s}} \leq f_2 n^s,
\]
where $f_2 = \frac{t^{t-s} 2^{t-s+1} b^{2(t-s)} q^{bt}}{s!}$.

By combining all the cases we obtain,
\[
\abs*{M_{i,j}} \leq \begin{cases}
    0, &\text{ if } i+j \leq 2(t-s-1),\\
    f_1 n^s, &\text{ if } i \leq s \text{ or } j \leq s,\\
    f_2 n^s, &\text{ if } i,j \geq s+1.
\end{cases}
\]
Since $f_2 \geq f_1$, we can unify the bound, and say that for all $i+j\geq 2(t-s-1)+1$,
\[
\abs*{M_{\leq t}} \leq \sum_{0\leq i,j\leq t} \abs*{M_{i,j}} \leq (t+1)^2 f_2 n^s - 1.
\]
Finally, the reconstruction degree of $\cC$ satisfies
\[
N(\cC) = 1 + \max_{\bx\neq \by \in \cC} \set{M_{\leq t}} \leq (t+1)^2 f_2 n^s.
\]
\end{proof}

One corollary of Theorem~\ref{thm:construction degree} gives us the order of the intersection of balls.

\begin{cor}\label{cor:Intersection}
Let $q,b,t,k$ be constant positive integers, with $q \geq 2$. If $\bx,\by\in\Sigma_q^n$ are such that $d_b(\bx,\by)\geq 2k-1$, then
\[
    \abs*{\Ball_{t,b}(\bx) \cap \Ball_{t,b}(\by)} = O(n^{t-k}).
\]
\end{cor}
\begin{proof}[Proof of~\cref{cor:Intersection}]
Take the $(k-1,b)$-burst-correcting code $\cC=\set{\bx,\by}$. Its reconstruction degree (minus one) is the size we are looking for. We then apply Theorem~\ref{thm:construction degree}.
\end{proof}

Another consequence of Theorem~\ref{thm:construction degree} is the following.

\begin{theorem}\label{thm:upper}
Let $q,b,t$ be constant positive integers, $q \geq 2$. Then, for all $0\leq s\leq t-1$,
\[
\ra(\ch(t,b),s) \leq  \ra(\ch(t-s-1,b),-1)\leq 2(t-s-1).
\]
\end{theorem}
\begin{proof}[Proof of~\cref{thm:upper}]
Simply combine Theorem~\ref{thm:construction degree} with Corollary~\ref{cor:Redu}, using the same logic as the proof of Theorem~\ref{thm:MARR for sub}.
\end{proof}

\subsection{Summary of \texorpdfstring{$\ra(\ch(t,b),s)$}{}}

We summarize our bounds on $\ra(\ch(t,b),s)$:

\begin{theorem}
\label{thm:MARR for burst}
Let $q,b,t$ be constant positive integers, $q \geq 2$. Then
\[
\ra(\ch(t,b),s) =  \ra(\ch(t-s-1,b),-1),
\]
for all $0\leq s\leq t-1$. If Conjecture~\ref{conj: burst} is true, we further have
\[
\ra(\ch(t,b),s) =  \ra(\ch(t-s-1,1),-1),
\]
for all $0\leq s\leq t-1$.
\end{theorem}
\begin{proof}[Proof of~\cref{thm:MARR for burst}]
Combine the lower and upper bounds of Theorem~\ref{thm:lower} and Theorem~\ref{thm:upper}, to obtain the desired result.
\end{proof}

Loosely speaking, Theorem~\ref{thm:MARR for burst} shows that we can use a burst-correcting code over a channel that introduces $s+1$ more burst errors than the code is designed to correct. When we do so, the code becomes a reconstruction code whose reconstruction degree is $O(n^{s})$.

\subsection{Improvement of the GV bound for burst-correcting codes}

While this section is devoted to reconstruction codes, the results we have obtained in it affect bounds on burst-correcting codes. In this subsection, we further refine the upper bound on the minimal redundancy of $(t,b)$-burst-correcting codes. While Corollary~\ref{cor:Redu} gives an upper bound of $(2t\log_q n)(1+o(1))$, the following theorem is slightly stronger. This is accomplished by harnessing Corollary~\ref{cor:Intersection} on the size of the intersection of error balls.

\begin{theorem}\label{thm:impoved upperred}
Let $q,b,t$ be constant positive integers, $q \geq 2$. Then there exists a $(t,b)$-burst-correcting code,  $\cC\subseteq\Sigma_q^n$, with
\[
r(\cC)\leq 2t\log_q n - \log_q \log n + O(1).\]
\end{theorem}

\begin{proof}[Proof of~\cref{thm:impoved upperred}]
We construct an auxiliary graph $G=(V(G),E(G))$ with vertex set $V(G) = \Sigma_q^n$ and edge set $E(G) = \set{(\bx,\by):\bx,\by\in \Sigma_q^n, d_b(\bx,\by) \leq 2t}$. Note that any independent set $\cC \subseteq \Sigma_q^n$ corresponds to a $(t,b)$-burst-correcting code.

For any two distinct $\bx,\by\in V(G)$, by~Theorem~\ref{thm:ball size} and Corollary~\ref{cor:Intersection}, we have the degree and joint degree satisfying
    \begin{align*}
        \gd(\bx) &= \abs*{\set*{\bz: d_b(\bx,\bz) \leq 2t}} = \abs*{\Ball_{2t,b}(\bx)} = \Theta(n^{2t}),\\
        \jd(\bx,\by) &= \abs*{\set*{\bz: d_b(\bx,\bz) \leq 2t, d_b(\by,\bz) \leq 2t}} = \abs*{\Ball_{2t,b}(\bx) \cap \Ball_{2t,b}(\by)} = O(n^{2t-1}).
    \end{align*}
That means, $\gdmax(G) = \Delta = \Theta(n^{2t})$, and $\jdmax(G) = O(n^{2t-1}) \leq \frac{\Delta}{(2\log \Delta)^7}$. By~Theorem~\ref{codegree}, for all sufficiently large $n$, we have an independent set in $G$ of size
    \[
    \alpha(G) \geq (1-o(1))\frac{q^n\log \Delta}{\Delta} = (1-o(1)) \frac{q^n 2t \log n}{n^{2t}}.
    \]
Thus the redundancy is at most $n-\log_q \alpha(G) = 2t \log_q n - \log_q \log n + O(1)$.
\end{proof}

\section{List-Reconstruction Codes in \texorpdfstring{$\ch(t,b)$}{}}
\label{sec:list}

Reviewing our progress thus far, we started by noting the obvious -- over the channel $\ch(t,b)$, we can uniquely decode using a single read by using a $(t,b)$-burst-correcting code. Trading code power for number of reads, in Section~\ref{sec:reconstruction code} we showed that over the same channel we can achieve unique decoding by using a $(t-s-1,b)$-burst-correcting code with $O(n^s)$ reads.

The goal of this section is to take the next step, which is to trade reads for list size (thereby, departing from unique decoding). Instead of using $O(n^s)$ reads (required for unique decoding), we use only $\Theta(n^{s-h})$ reads. We determine $\La(h)$ for some non-negative integer $0\leq h\leq s$. When $h = 0$, Theorem~\ref{thm:construction degree} implies that a list of size $1$ is enough. Hence, we always assume $1 \leq h \leq s$. 

For fixed positive integers $w$ and $r$, Chee and Ling~\cite{chee2006constructions} provided the Johnson bound for $(r,1)$-burst-correcting codes $\cS \subseteq \Sigma_q^n$, where $\cS \subseteq \Ball_{w,1} (\mathbf{0})$, i.e., where the codewords have weights no more than $w$. For sufficiently large $n$, they showed that $\abs{\cS} = O(n^{w-r})$. Recently, Liu and Shangguan~\cite[Theorem 1.4]{liu2025approximate} demonstrated the existence of such a code $\cS$ with $\abs{\cS} = \Omega(n^{w-r})$. We will generalize these results to the channel $\ch(r,b)$ for $b\geq 1$.

\begin{theorem}\label{thm:upper johnson}
Let $w,r,b,n,q$ be positive integers, with $w\geq r$, and $q\geq 2$. Then for any $(r,b)$-burst-correcting code, $\cS \subseteq \Sigma_q^n$, such that $\cS \subseteq \Ball_{w,b} (\mathbf{0})$, we have
\[
\abs*{\cS}  \leq (w+1)(q^b-1)^{w-r}n^{w-r}.
\]
\end{theorem}

\begin{proof}[Proof of~\cref{thm:upper johnson}]
Let \( A(n,d,k) \) denote the maximum size of a code in \( \Sigma_q^n \), with weight \( k \), and minimum \( b^\leq \)-burst distance at least \( d \). Clearly,  
\[
\abs*{\cS} \leq \sum_{k=0}^{w} A(n, 2r+1, k).
\]  
For \( k \leq r \), we have \( A(n, 2r+1, k) = 1 \). 

We turn to handle the case $k\geq r+1$. Let \( \cC \subseteq \Sigma_q^n \) be a code of size \( A(n, 2r+1, k) \), with weight \( k\geq r+1 \), and minimum \( b^\leq \)-burst distance at least \( 2r+1 \). For each \( i \in [n] \) and \( \bu \in \Sigma_q^b \setminus \set{ \mathbf{0} } \), define  
\[
\cC_{i,\bu} = \set*{ \bx \in \cC : \bx[i, i+b-1] = \bu, \wt_b(\bx\overline{[i, i+b-1]}) = k-1 },
\]  
and  
\[
\cC_{i,\bu}^* = \set*{ \bx\overline{[i, i+b-1]} : \bx \in \cC_{i,\bu} },
\]  
(see the definitions in Section~\ref{sec:prelim}). By construction, we have  
\[
\abs*{\cC_{i,\bu}} = \abs*{\cC_{i,\bu}^*}\leq A(n-b, 2r+1, k-1).
\]  
On the other hand, since  
\[
\cC = \bigcup_{i\in [n],\bu \in \Sigma_q^b \setminus \set{ \mathbf{0} }} \cC_{i,\bu},
\]  
we conclude that  
\[
A(n, 2r+1, k) \leq n (q^b-1) A(n-b, 2r+1, k-1).
\]
By recursively using the above inequality $k-r$ times, we obtain that 
\[
A(n,2r+1,k) \leq n^{k-r} (q^b-1)^{k-r} A(n-(k-r)b,2r+1,r) = n^{k-r} (q^b-1)^{k-r}.
\]
Then we have 
\[\abs*{\cS}\leq \sum_{k=0}^w A(n,2r+1,k) \leq (w+1)(q^b-1)^{w-r}n^{w-r} = O(n^{w-r}),\]
as claimed.
\end{proof}
    
\begin{theorem}\label{thm:johnson}
For any fixed positive integers $w,r,b$, there exists a sequence of $(r,b)$-burst-correcting code, $\cC_n \subseteq \Sigma_q^n$, with $\cC_n \subseteq \Ball_{w,b} (\mathbf{0})$, such that $\abs{\cC_n} = \Omega(n^{w-r})$.
\end{theorem}

\begin{proof}[Proof of~\cref{thm:johnson}]
Let $\cI$ be the collection of $w$-subsets of $[n]$ defined by
\[
\cI = \set*{\set*{i_1,\ldots,i_w}\subseteq[n] : 1 \leq i_1 < \cdots < i_w \leq n, i_{j} + 3b \leq i_{j+1} \text{ for $j\in [w]$, where } i_{w+1} = i_1+n}.
\]
We define a $w$-uniform hypergraph $\cG_q(n,w)$ with vertices and edges,
\begin{align*}
    V(\cG_q(n,w))&=[n], &
    E(\cG_q(n,w))&=\cI.
\end{align*}
With each edge $I=\set{i_1,\dots,i_w} \in E(\cG_q(n,w))$ we associate a unique length-$n$ sequence $\bx_I \in \Ball_{w,b}(\mathbf{0})$ all of whose entries are zeros, except $\bx[i_j,i_j+b-1]$ is an all-one vector, for all $j\in[w]$. The number of edges in this hypergraph is $\abs{E(\cG_q(n,w))} \geq \binom{n-(3b-1)w}{w} = \Omega(n^w)$. The following claim is crucial in the proof.

\begin{claim}\label{claim:IntersectionPatterns}
For any distinct edges $I = \set{i_1,\ldots, i_w},J = \set{j_1,\ldots, j_w}\in E(\cG_q(n,w))$, $d_b(\bx_I, \bx_J) = 2w - 2\abs{I \cap J}$. 
        \end{claim} 
\begin{poc}
Let \( U_k = [i_k, i_k+b-1] \) and \( V_k = [j_k, j_k+b-1] \) denote cyclic intervals for \( k \in [w] \). Then, the support sets satisfy  
\[
\supp(\bx_I) = \bigcup_{k=1}^w U_k, \quad \supp(\bx_J) = \bigcup_{k=1}^w V_k.
\]  
Each \( U_k \) intersects with at most one interval from \( \set{V_1, \ldots, V_k} \), since the gap (see Section~\ref{sec:prelim}) between any two intervals in \( \set{V_1, \ldots, V_k} \) is strictly greater than \( 2b \), ensuring that they remain disjoint. Similarly, each \( V_k \) intersects with at most one interval from \( \set{U_1, \ldots, U_k} \).

\begin{figure}
    \centering
    \includegraphics[width=0.99\linewidth]{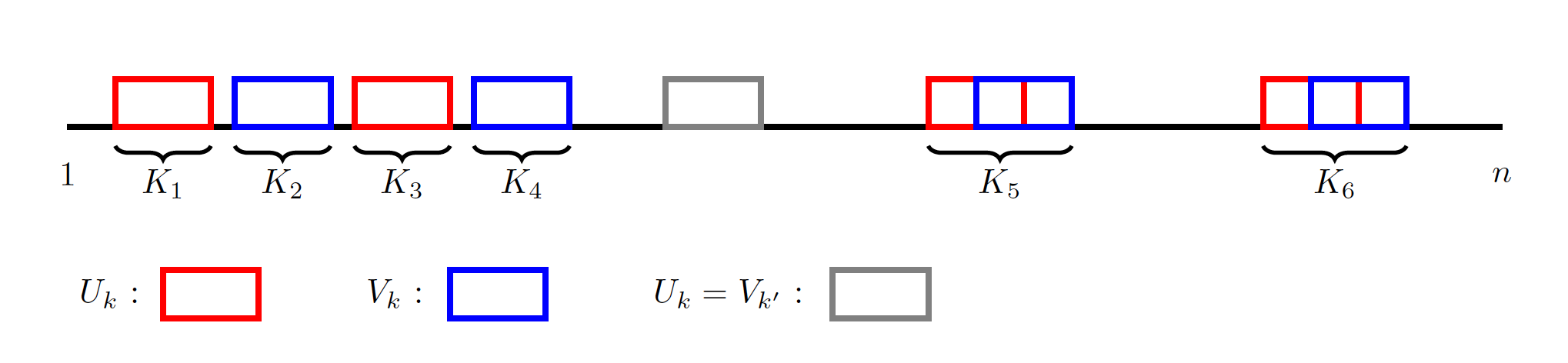}
    \caption{There are three cases of the intersections between $\set{U_1,\ldots,U_k}$ and $\set{V_1,\ldots,V_k}$ in $\supp(\bx_I - \bx_J)$.}
    \label{fig:supp}
\end{figure}

Let us denote by $m_1$ the number of indices $k\in[w]$ such that $U_k$ intersects some $V_\ell$, $\ell\in[w]$ but not fully, i.e., $U_k\cap V_\ell\neq\emptyset$ and $U_k\neq V_\ell$. Let us also denote by $m_2$ the number of indices $k\in[w]$ such that $U_k$ is disjoint from all $V_\ell$, $\ell\in[w]$. In Fig.~\ref{fig:supp} an example is shown. Note that the gray bursts in the figure completely intersect, and are of no interest to us. By definition we have $ m_1 + m_2 = w - \abs{I \cap J} $.

We now consider the vector $\bx_I-\bx_J$. There exist $ m_1 + 2m_2 $ disjoint consecutive cyclic intervals (some perhaps longer than $b$), we denote by $ K_1, \ldots, K_{m_1 + 2m_2} $, such that the support $ \supp(\bx_I - \bx_J) $ is contained within their union. Notably, the endpoints of each interval $ K_i $ lie within this support for all $ i \in [m_1 + 2m_2] $.
There are three possible types of $ K_i $:
        \begin{enumerate}
            \item Type I: $K_i = U_k = [i_k,i_k+b-1]$ for some $k\in [w]$ (created by $U_k$ that does not intersect any $V_\ell$),
            \item Type II: $K_i = V_\ell = [j_\ell,j_\ell+b-1]$ for some $\ell\in [w]$ (created by $V_\ell$ that does not intersect any $U_k$),
            \item Type III: $K_i = U_k \cup V_\ell$ for some $U_k \cap V_\ell \neq \emptyset, U_k \neq V_\ell$, where 
            \[ K_i = 
            \begin{cases}
                [i_k,j_\ell+b-1], & \text{when }j_\ell \in [i_k+1,i_k+b-1], \\
                [j_\ell,i_k+b-1],& \text{when }i_k \in [j_\ell+1,j_\ell+b-1].
            \end{cases}
            \]
        \end{enumerate}
We further index the $K_i$ such that their starting point is in ascending order.

For a pair of consecutive intervals \( (K_i, K_{i+1}) \), there are \( 3 \times 3 = 9 \) possible combinations of interval types. We now show that the gap between any such pair is at least \( b \) for the following seven cases:

\textbf{Case 1:} \( (\text{Type}(K_i), \text{Type}(K_{i+1})) = (\text{Type I, Type I}) \) or \( (\text{Type II, Type II}) \)  

By definition, the gap between these intervals is at least \( 3b \), since \( K_i \) and \( K_{i+1} \) are two disjoint intervals from either \( \set{U_1, \ldots, U_w} \) or \( \set{V_1, \ldots, V_w} \).  

\textbf{Case 2:} \( (\text{Type}(K_i), \text{Type}(K_{i+1})) = (\text{Type I, Type III}) \) or \( (\text{Type III, Type I}) \) or \( (\text{Type II, Type III}) \) or \( (\text{Type III, Type II}) \)  

We consider the case \( (\text{Type}(K_i), \text{Type}(K_{i+1})) = (\text{Type I, Type III}) \), as the proofs for the other cases are analogous. Without loss of generality, assume  
\[
K_i = [i_s, i_s+b-1], \quad K_{i+1} = [j_\ell, i_k+b-1],
\]  
where \( i_s < i_k \) and \( i_k \in [j_\ell+1, j_\ell+b-1] \). The gap is given by  
\[
\min \set*{ j_\ell - i_s - b, i_s + n - i_k - b }.
\]  
Since  
\[
j_\ell - i_s - b \geq (i_k - b + 1) - i_s - b \geq 3b - 2b + 1 = b+1,
\]  
and  
\[
i_s + n - i_k - b \geq 3b - b = 2b,
\]  
we conclude that the gap is at least \( b+1 \).  

\textbf{Case 3:} \( (\text{Type}(K_i), \text{Type}(K_{i+1})) = (\text{Type III, Type III}) \)

By symmetry, it suffices to consider the case where  
\[
K_i = [j_s, i_t + b - 1], \quad K_{i+1} = [j_\ell, i_k + b - 1],
\]  
with \( i_t < i_k \), \( i_t \in [j_s+1, j_s+b-1] \), and \( i_k \in [j_\ell+1, j_\ell+b-1] \). The gap is given by  
\[
\min \set*{ j_\ell - i_t - b, j_s + n - i_k - b } \geq b+1.
\]  

To complete the proof of the claim, let \( m = d_b(\bx_I, \bx_J) \), and let \( L_1, \dots, L_m \) be pairwise disjoint cyclic intervals of length at most \( b \) such that their union contains \( \supp(\bx_I - \bx_J) \). Clearly, we have  
\[
m \leq 2m_1 + 2m_2,
\]
since we can cover $K_i$ of type I or II with a single interval $L_j$, and those of type III by two intervals.

In the other direction, for each \( K_j \) of Type III, the interval \( L_i \cap K_j \neq \emptyset \) cannot intersect any other \( K_\ell \) for \( \ell \neq j \), as the gap between \( K_j \) and \( K_\ell \) exceeds the length of \( L_i \). To cover both endpoints of \( K_j \), at least two cyclic intervals of length at most \( b \) from \( \set{L_1, \dots, L_m} \) are required. 
For the remaining \( 2m_2 \) intervals \( K_j \) of Type I or Type II, we similarly require at least \( 2m_2 \) cyclic intervals of length at most \( b \) to cover them. Thus,  
\[
m \geq 2m_1 + 2m_2,
\]  
which implies  
\[
m = 2(m_1 + m_2) = 2w - 2 \abs*{I \cap J}.
\]  
This completes the proof of the claim.
\end{poc}

Now we apply Lemma~\ref{lem:kahn} to prove the theorem. Let \( t = w - r \), and define a hypergraph \( \cH \) with vertices and edges
\begin{align*}
V(\cH) &= \binom{V(\cG_q(n,w))}{t}, &
E(\cH) &= \set*{ \binom{e_I}{t} : e_I \in E(\cG_q(n,w)) }.
\end{align*}
Then, \( \cH \) is a \( \binom{w}{t} \)-uniform hypergraph, and its number of edges satisfies  
\[
\abs*{E(\cH)} = \abs*{E(\cG_q(n,w))} = \Omega(n^w).
\]  

Assume \( u, v \in V(\cH) \) are two distinct vertices, say  
\[
u = \set*{ i_1, \ldots, i_t }, \quad v = \set*{ j_1, \ldots, j_t }.
\]  
The degree of \( u \) is given by  
\[
\gd(u) = \abs*{ \set*{ I \in \cI : \set*{i_1, \ldots, i_t} \subseteq I } } = O(n^{w-t}).
\]  
Moreover, the maximum degree, \( \gdmax(\cH)=\Omega(n^{w-t}) \) as evident, for example, from the vertex $\set*{1, 3b+1, \ldots, 3b(t-1)+1}$.
Additionally, we can see the codegree of $u$ and $v$ is 
\[
\cod(u,v) = \abs*{I\in \cI : \set{i_1,\ldots,i_t,j_1,\ldots,j_t} \subseteq I} \leq \binom{n}{w-\abs{\set{i_1,\ldots,i_t,j_1,\ldots,j_t}}} =O(n^{w-t-1}).
\]
        Therefore, we obtain  
\begin{align*}
\gdmax(\cH) &= \Theta(n^{w-t}), &
\codmax(\cH) &= O(n^{w-t-1}) = o(\Delta(\cH)).
\end{align*}
Applying Lemma~\ref{lem:kahn}, we conclude that  
\[
\nu(\cH) \geq (1 - o(1)) \frac{\abs*{E(\cH)}}{\gdmax(\cH)} = \Omega(n^t) = \Omega(n^{w-r}).
\]  

Let $S$ be a largest matching in $\cH$, i.e., $\abs{S}=\nu(\cH)$. By our construction, there exists $\cS\subseteq\cI$ such that
\[
S = \set*{ \binom{\be_I}{t} : I \in \cS}.
\]  
We construct a code  
\[
\cC_S = \set*{ \bx_I : I \in \cS }.
\]  
By the definition of \( \cI \), it follows that  
$\cS \subseteq \Ball_{w,b}(\mathbf{0})$.
Since $S$ is a matching, for each pair of distinct \( \bx_I, \bx_J \in \cC_S \), it is clear that $\abs{\be_I \cap \be_J} \leq t-1$. Thus, for any distinct $\bx_I,\bx_J\in\cC_S$, by Claim~\ref{claim:IntersectionPatterns}, we obtain  
\[
d_b(\bx_I, \bx_J) = 2w - 2\abs{I \cap J} = 2w - 2\abs{\be_I \cap \be_J} \geq 2w - 2(t-1) = 2r + 2 > 2r + 1.
\]  
Thus, \( \cC_S \) forms an \( (r,b) \)-burst-correcting code with size \( \abs{S}=\nu(\cH)=\Omega(n^{w-r}) \).
\end{proof}
    
By applying Theorem~\ref{thm:johnson}, we can obtain a lower bound on the asymptotic growth rate of the list size of reconstruction codes.

\begin{theorem}\label{thm:lower list recon}
Let $q,b,t,s,h$ be constant positive integers, $q \geq 2$, and $t\geq s\geq h\geq 1$. Then, when using $(t-s-1,b)$-burst-correcting codes of length $n$, over $\ch(t,b)$, with $N(n)$ reads, $\Ord_b(N(n))=s-h$, the asymptotic growth rate of the list size satisfies
\[
    \La(h) \geq h,
\]
namely, the list size grows at least as $n^{(1-o(1))h}$.
\end{theorem}

\begin{proof}[Proof of~\cref{thm:lower list recon}]
We construct a sequence of codes $\cC_n\subseteq\Sigma_q^n$ in the following way. Setting \( r = t-s-1 \) and \( w = t-s+h-1 \) in Theorem~\ref{thm:johnson}, we can construct a \( (t-s-1,b) \)-burst-correcting code \( \cC_n \) such that 
$\cC_n \subseteq \Ball_{t-s+h-1,b}(\mathbf{0})$,
and $\abs{\cC_n} = \Omega(n^h)$. Since each codeword \( \bx \in \cC_n \) has \( b^\leq \)-weight at most \( t-s+h-1 \), it follows that  
\[
\Ball_{s-h+1,b}(\mathbf{0}) \subseteq \Ball_{t,b}(\bx).
\]  
Thus,  
\[
\Ball_{s-h+1,b}(\mathbf{0}) \subseteq \bigcap_{\bx \in \cC_n} \Ball_{t,b}(\bx).
\]  
It now follows that
\[
\abs*{ \bigcap_{\bx\in\cC_n} \Ball_{t,b}(\bx) } \geq \abs*{\Ball_{s-h+1,b}(\mathbf{0})} = \Omega(n^{s-h+1}).
\]  
This means that unique reconstruction of $\cC_n$ requires $\Omega(n^{s-h+1})$ reads. Conversely, let $N(n)$ be any sequence such that $\Ord_b(N(n))=s-h$, which in particular, implies $N(n)=o(n^{s-h+1})$. Thus, for all sufficiently large $n$, using $N(n)$ reads for decoding $\cC_n$, the list size must include all of $\cC_n$, i.e., for all sufficiently large $n$,
\[
L(n,\cC_n,\ch(t,b),N(n)) = \abs*{\cC_n}=\Omega(n^h),
\]
and therefore
\[ L(s,N(n)) \geq L(n,\cC,\ch(t,b),N(n)) = \Omega(n^h).\]
Hence,
\[
\liminf_{n\to\infty} \frac{\log_q L(s,N(n))}{\log_q n} \geq h.
\]
Since this holds for any $N(n)$ such that $\Ord_b(N(n))=s-h$, it follows that
\[
\La(h) =\inf_{N(n):\Ord_{b}(N(n))=s-h}\liminf_{n\to\infty} \frac{\log_q L(s,N)}{\log_q n} \geq h.
\]
\end{proof}

We turn to consider an upper bound on $\La(h)$. Let us first consider the substitution channel $\ch(t,1)$. Here, the strategy is finding a pair of reads $\by_i,\by_j$ with large distance among the $N$ reads, which induces $p+1$ codewords located in the intersection $\Ball_{t,1}(\by_i) \cap \Ball_{t,1}(\by_j)$. Then, the classical Johnson bound yields an upper bound.

Several works may help us in this strategy. For example, a fundamental result of Kleitman~\cite{1966Kleitman} states that for $n\ge 2d+1$, if any pair of binary sequences in a set  $\cF\subseteq \Sigma_2^n$ has Hamming distance at most $2d$, then $\abs{\cF}\le \sum_{i=0}^{d}\binom{n}{i}$. Also see~\cite{ahlswede1992diametric,ahlswede1998diametric,frankl1980erdos,2023Nonshadow,huang2020subsets}. In the context of larger alphabets and under the condition of bounded Hamming distance, Ahlswede \textit{et al}.~\cite{ahlswede1992diametric} generalized  Kleitman's Theorem as follows.

\begin{theorem}[{{\cite[Theorem 7]{ahlswede1992diametric}}}]
\label{thm:Ahl}
Let $q,d$ be fixed integers, $n\geq (q-1)^{2d-1}+2d$, and let $\cA\subseteq\Sigma_q^n$. If for any two $\bx,\by \in \cA$, the Hamming distance $d_1(\bx,\by) \leq 2d$, then we have  
\[
      \abs*{\cA} \leq \sum_{i=0}^d \binom{n}{i}(q-1)^i = O(n^d).
      \]
\end{theorem}

Therefore, an interesting question arises: Under which distance definitions and for different alphabets can we obtain results analogous to Kleitman's diametric inequality? Next, we will prove that under the burst distance, a similar inequality holds, which may be of independent interest. Before mentioning our result, we first introduce some necessary definitions.

\begin{definition}
For any set $\cA \subseteq \Sigma_q^n$, $\bx = (x_1,\ldots, x_n) \in \cA$, $i\in [n]$, and $a\in \Sigma_q$, we define the shifting function $S_{i,a}$ of $\bx$ to be 
    \[
    S_{i,a}(\bx) = \begin{cases}
        \bx' = (x_1,\ldots, x_{i-1}, 0, x_{i+1},\ldots,x_n), &\text{ if } x_i = a \text{ and } \bx' \notin \cA,\\
        \bx, &\text{ otherwise}.
    \end{cases}
    \]
Define $S_{i,a}(\cA) = \set{S_{i,a}(\bx): \bx \in \cA}$.
\end{definition}

It is clear that $\abs{S_{i,a}(\cA)} = \abs{\cA}$. For any $\cA \subseteq \Sigma_q^n$ we define the \emph{diameter} of $\cA$ as the largest distance between a pair of its elements. Now, if $\cA$ has diameter $D$, we show that $S_{i,a}(\cA)$ also has diameter at most $D$, namely, the shifting of $\cA$ does not increase its diameter.

\begin{prop}\label{prop:NoChangeShift}
Let $\cA\subseteq \Sigma_{q}^{n}$ be a set with $b^{\leq}$-diameter $D$. Then for any $\bx,\by\in\cA$,  we have $d_b(S_{i,a}(\bx),S_{i,a}(\by))\le D$.
\end{prop}
\begin{proof}[Proof of~\cref{prop:NoChangeShift}]
For any two \( \bx, \by \in \cA \), we bound \( d_b(S_{i,a}(\bx), S_{i,a}(\by)) \) by considering four cases.  

\textbf{Case 1:} If \( S_{i,a}(\bx) = \bx \) and \( S_{i,a}(\by) = \by \), then it is clear that  
\[
d_b(S_{i,a}(\bx), S_{i,a}(\by)) = d_b(\bx, \by) \leq D.
\]  

\textbf{Case 2:} If \( S_{i,a}(\bx) = \bx \) and \( S_{i,a}(\by) \neq \by \), then by definition, we have \( \by[i] = a \neq 0 \) and \( S_{i,a}(\by)[i] = 0 \). If \( \bx[i] \neq \by[i] \), then  
\[
\supp(S_{i,a}(\bx)- S_{i,a}(\by)) \subseteq \supp(\bx - \by),
\]  
which implies that  
\[
d_b(S_{i,a}(\bx), S_{i,a}(\by)) \leq d_b(\bx, \by).
\]  
Otherwise, when \( \bx[i] = \by[i] \), we have \( \bx[i] = a \). Since \( S_{i,a}(\bx) = \bx \), it follows that the sequence  
$\bx' = (x_1, \ldots, x_{i-1}, 0, x_{i+1}, \ldots, x_n)$  
belongs to \( \cA \). Thus,  
\[
\supp(S_{i,a}(\bx)- S_{i,a}(\by)) = \supp(\bx' - \by),
\]  
which yields  
\[
d_b(S_{i,a}(\bx), S_{i,a}(\by)) = d_b(\bx', \by) \leq D.
\]  

\textbf{Case 3:} If \( S_{i,a}(\bx) \neq \bx \) and \( S_{i,a}(\by) = \by \), we can apply the same argument as in Case 2, and we omit the repeated details.  

\textbf{Case 4:} If \( S_{i,a}(\bx) \neq \bx \) and \( S_{i,a}(\by) \neq \by \), then by definition, we have \( \bx[i] = a = \by[i] \) and \( S_{i,a}(\bx)[i] = 0 = S_{i,a}(\by)[i] \). Therefore,  
\[
S_{i,a}(\bx) - S_{i,a}(\by) = \bx - \by,
\]  
which implies that  
\[
d_b(S_{i,a}(\bx), S_{i,a}(\by)) = d_b(\bx, \by) \leq D.
\]  
\end{proof}

Ahlswede \textit{et al}.~\cite{ahlswede1998diametric} showed that any set $\cA \subseteq \Sigma_q^n$ can be transformed into a set $\cA'$ by finitely repeating shifting on $\cA$, such that $\cA'$ is a fixed point, namely,
$S_{i,a}(\cA') = \cA'$,
for all $i\in [n]$ and $a\in \Sigma_q$. In other words, $S_{i,a}(\bx) = \bx$ for $\bx \in \cA'$. We shall use a similar approach. In the following, we say $\cA$ is isomorphic to $\cA'$ if we can obtain $\cA'$ from $\cA$ using a sequence of shift operations.
 
\begin{theorem}\label{thm:burst diameter}
Let $q,d,b$ be fixed positive integers, $q\geq 2$. For all sufficiently large $n$, if $\cA\subseteq\Sigma_{q}^{n}$ has $b^{\leq}$-diameter $2d$, then
      \[
      \abs*{\cA} \leq \abs{\Ball_{d,b}(\boldsymbol{0})},
      \]
and equality holds if and only if $\cA$ is isomorphic to $\Ball_{d,b}(\boldsymbol{0})$.
\end{theorem}

\begin{proof}[Proof of~\cref{thm:burst diameter}]
By definition and Proposition~\ref{prop:NoChangeShift}, the shifting operation \( S_{i,a}(\cA) \) neither alters the size of \( \cA \) nor increases its diameter. By repeating shifts we must reach a fixed point. Thus, we may assume without loss of generality that \( \cA \) is that fixed point, namely, it satisfies  
\[
S_{i,a}(\cA) = \cA \quad \text{for all } i \in [n] \text{ and } a \in \Sigma_q.
\]  

For any \( \bx = (x_1, \ldots, x_n) \in \cA \) with \( x_i \neq 0 \) for some \( i \in [n] \), we observe that 
$S_{i,x_i}(\bx) = \bx$.
This implies that the vector  
$\bx' = (x_1, \ldots, x_{i-1}, 0, x_{i+1}, \ldots, x_n)$,  
obtained by setting the \( i \)-th entry of \( \bx \) to zero, also belongs to \( \cA \). Thus, for any \( \bx \in \cA \), any vector obtained by flipping an arbitrary subset of nonzero entries in \( \bx \) to zero remains in \( \cA \), and in particular, \(\mathbf{0} \in \cA\). Additionally, \( \wt_b(\bx) \leq 2d \) for all \( \bx \in \cA \). We analyze the size of \( \cA \) by considering the following cases.

\textbf{Case 1:} If there exists \( \bv \in \cA \) with \( \wt_b(\bv) = k \), where \( k \in [d+1, 2d] \), let the bursts of \( \bv \) be covered by \( k \) disjoint intervals of length at most \( b \), denoted by \( I_1, \ldots, I_k \). Define the extended intervals \( I_i^* \) as the corresponding \( (1, b) \)-extensions of \( I_i \) (see definition at~\eqref{eq:ext}), and let  
\[
I = \bigcup_{i=1}^k I_i, \quad I^* = \bigcup_{i=1}^k I_i^*, \quad \text{and} \quad \ell = \abs*{I^*},
\]  
hence, \( \ell \leq 3bk \leq 6bd \).  

For any \( \bu \in \Sigma_q^\ell \), define  
\[
\cA_{\bu} = \set*{ \bx \in \cA : \bx[I^*] = \bu }.
\]  
Clearly,  
\[
\cA = \bigcup_{\bu \in \Sigma_q^\ell} \cA_{\bu},
\]  
is a partition of $\cA$.

For any \( \bu \in \Sigma_q^\ell \) and \( \bw \in \cA_{\bu} \), define \( \bw' \) as follows:  
\[
\bw'[j] =  
\begin{cases}  
\bw[j], & j \in [n] \setminus I^*, \\  
0, & j \in I^*.
\end{cases}
\]  
By our earlier observation, \( \bw' \in \cA \). Since the gap between each burst of \( \bw' \) and each burst of \( \bv \) is greater than \( b \), we obtain  
\[
d_b(\bw', \bv) \geq d_b(\bw', \mathbf{0}) + d_b(\mathbf{0}, \bv) = \wt_b(\bw') + k,
\]  
where the inequality follows from the definition of the \( b^\leq \)-burst distance. Thus,  
\[
\wt_b(\bw') \leq d_b(\bw', \bv) - k \leq 2d - k \leq d - 1.
\]  
This implies that $\cA_{\bu}\subseteq \Ball_{d-1,b}(\mathbf{0})$, 
therefore $\set{\bm{w}':\bm{w}\in \cA_{\bm{u}}} \subseteq \Ball_{d-1,b}(\mathbf{0})$. Moreover, since there is a one-to-one correspondence between $\bm{w}'$ and $\bm{w}$ in $\cA_{\bm{u}}$, we get
\[
\abs*{\cA_{\bu}} \leq \abs*{\Ball_{d-1, b}(\mathbf{0})}.
\]  
It follows that,  
\[
\abs*{\cA} = \sum_{\bu \in \Sigma_q^\ell} \abs*{\cA_{\bu}} \leq q^\ell \abs*{\Ball_{d-1,b}(\mathbf{0})}.
\]  
By the asymptotic bound on the ball size from Theorem~\ref{thm:ball size}, we know that for all sufficiently large $n$,
\[
q^\ell \abs*{\Ball_{d-1,b}(\mathbf{0})} < \abs*{\Ball_{d,b}(\mathbf{0})}.
\]  
Consequently, in this case we have 
\[
\abs*{\cA} < \abs*{\Ball_{d,b}(\mathbf{0})}.
\]  

\textbf{Case 2:}
If every element \( \bv \in \cA \) satisfies \( \wt_b(\bv) \leq d \), then by the definition of the \( b^\leq \)-burst distance, we have  
$\cA \subseteq \Ball_{d,b}(\mathbf{0})$.
Consequently, the size of \( \cA \) satisfies  
\[
\abs*{\cA} \leq \abs*{\Ball_{d,b}(\mathbf{0})}.
\]  
Moreover, equality holds if and only if \( \cA \) is isomorphic to \( \Ball_{d,b}(\mathbf{0}) \).
\end{proof}

We are now at a position to prove the upper bound on $\La(h)$.

\begin{theorem}~\label{thm:upper list recon}
Let $q,b,t,s,h$ be constant positive integers, $q \geq 2$, and $t\geq s\geq h\geq 1$. Assume $\cC\subseteq \Sigma_q^n$ is a $(t-s-1,b)$-burst-correcting code. Then for $N\geq \abs{\Ball_{s-h,b}(\mathbf{0})}+1$ reads, the list size for $\cC$ satisfies,
\[
L(n,\cC,\ch(t,b),N) \leq t q^{b(h+10t)} n^h.
\]
\end{theorem}

\begin{proof}[Proof of~\cref{thm:upper list recon}]
Assume the codeword $\bx\in\cC$ was transmitted, and $N$ reads were received, $Y = \set{\by_1, \ldots, \by_{N}} \subseteq \Ball_{t,b}(\bx)$. Using these reads, we list-reconstruct 
\[
\cC\cap \bigcap_{k=1}^{N} \Ball_{t,b}(\by_k) = \set*{\bx_1, \ldots, \bx_{p(Y)}} \subseteq \bigcap_{k=1}^{N} \Ball_{t,b}(\by_k).
\]

Since $N\geq \abs{\Ball_{s-h,b}(\mathbf{0})}+1$, by Theorem~\ref{thm:burst diameter}, there exists a pair \( \by_i, \by_j \) such that  
\[
d_b(\by_i, \by_j) \geq 2(s-h) + 1.
\]  
Define \( \bx'_m = \bx_m - \by_i \) for each \( m \in [p(Y)] \). Then, we obtain  
\[
\set*{ \bx_1', \ldots, \bx_{p(Y)}' } \subseteq \bigcap_{k=1}^{N} \Ball_{t,b}(\by_k - \by_i) \subseteq \Ball_{t,b}(\mathbf{0}) \cap \Ball_{t,b}(\by_j - \by_i).
\]
Consider \( \by_j - \by_i \) as burst errors with \( k \) bursts, \( d_b(\by_i,\by_j)= k \geq 2(s-h) + 1 \), where the bursts are contained in pairwise disjoint cyclic intervals \( I_1, \ldots, I_k \). For each \( i \in [k] \), define \( I_i^{**} \) as the \( (2, b) \)-extension of \( I_i \). Set  
\[
I = \bigcup_{i=1}^k I_i, \quad I^{**} = \bigcup_{i=1}^k I_i^{**}, \quad \text{and} \quad
\ell = \abs*{I^{**}}.
\]  
It follows that  $\ell \leq 5bk \leq 10bt$.

Let us consider some arbitrary $\bx'_m$, $m\in[p(Y)]$. Since \( \bx'_m \in \Ball_{t,b}(\mathbf{0}) \cap \Ball_{t,b}(\by_j - \by_i) \), denote $\wt_b(\bx'_m)=w\leq t$. Assume that the bursts of \( \bx'_m \) are contained in pairwise disjoint intervals \( J_1, \ldots, J_w \). Additionally, we can decompose $\bx'_m$ into the sum of three sequences $\bv$, $\bw$, and $\bz$ as follows:
\[
     \bv[j] = \begin{cases}
         \bx'_m[j], &j\in [n]\setminus I^{**},\\
         0, &j\in I^{**},
     \end{cases} 
     \ 
     \bw[j] = \begin{cases}
         \bx'_m[j], &j\in I^{**}\setminus I,\\
         0, &j\in [n]\setminus I^{**} \cup I,
     \end{cases} 
     \ 
     \bz[j] = \begin{cases}
         \bx'_m[j], &j\in I,\\
         0, &j\in [n]\setminus I,
     \end{cases}
     \]
     and $\bx'_m = \bv + \bw + \bz$.
     
     \begin{claim}\label{claim:part weight}
         $\wt_b(\bv) \leq t-s+h-1$.
     \end{claim}
     
\begin{poc}
Assume to the contrary that \( \wt_b(\bv) = r \geq t-s+h \). Let \( K_1, \ldots, K_r \) be disjoint intervals, each containing a distinct burst of \( \bv \). Define \( K_i^* \) as the \( (1, b) \)-extension of \( K_i \), and set  
\[
K = \bigcup_{i=1}^r K_i, \quad K^* = \bigcup_{i=1}^r K_i^*.
\]  
By construction, \( K \) must intersect at least \( t-s+h \) intervals among \( J_1, \ldots, J_w \); otherwise, \( \bv \) would consist of at most \( t-s+h-1 \) bursts, contradicting our assumption that \( \wt_b(\bv) = r \geq t-s+h\). Expanding each \( K_i \) to \( K_i^* \), the set \( K^* \) captures all intervals \( J_1, \ldots, J_w \) that intersect with \( K \). Consequently, \( K^* \) contains at least \( t-s+h \) intervals from \( J_1, \ldots, J_w \).

By definition, the gap between \( K_i^* \) and \( I_j \) is at least \( 2b-b = b \) for all \( i \in [r] \) and \( j \in [k] \). Thus, \( I \subseteq [n] \setminus K^* \), implying that at most \( w - (t-s+h) \) cyclic intervals from \( J_1, \ldots, J_w \) remain in \( I \). Therefore,  
\[
\wt_b(\bz) \leq w - (t-s+h) \leq t - (t-s+h) = s-h.
\]  
Applying the triangle inequality,  
\[
\wt_b(\bz - (\by_j - \by_i)) \geq \wt_b(\by_j - \by_i) - \wt_b(\bz) \geq k - s + h \geq s-h+1.
\]  
Now, considering the distance between \( \bx'_m \) and \( \by_j - \by_i \), we have  
\[
\bx'_m - (\by_j - \by_i) = \bv + \bw + (\bz - (\by_j - \by_i)).
\] Therefore 
    \[
    d_b((\by_j-\by_i), \bx'_m) = \wt_b(\bx'_m - (\by_j-\by_i)) \geq \wt_b(\bv) + \wt_b(\bz-(\by_j-\by_i)) 
    \]
    \[
    = r + \wt_b(\bz-(\by_j-\by_i)) \geq t+1,
    \] 
    where the first inequality comes from the property that the gap between each $I_i$ and each $K_j$ is at least $b$ for $i\in [k], j\in [r]$, which yields that any burst is contained in either $I$ or $K$. This is a contradiction to the assumption that $\bx'_m\subseteq \Ball_{t,b}(\by_{i}-\by_{j})$. 
    \end{poc}

By definition, the set \( \cS = \set{ \bx'_1, \ldots, \bx'_{p(Y)} } \) forms a \( (t-s-1, b) \)-burst-correcting code. Partitioning \( \cS \) into \( q^{\ell} \) disjoint classes, we define  
\[
\cS = \bigcup_{\bu \in \Sigma_q^\ell} \cS_{\bu}, \quad \text{where} \quad \cS_{\bu} = \set*{ \bx'_m \in \cS : \bx'_m[I^{**}] = \bu }.
\]
It is clear that each $\cS_{\bu}$ is also a $(t-s-1,b)$-burst-correcting code. For each class, consider the set
\[
\cS_{\bu}^* = \set*{ \by_{\bx'_m} \in \Sigma_q^n : \by_{\bx'_m}[I^{**}] = 0^\ell, \, \by_{\bx'_m}\overline{[I^{**}]} = \bx'_m\overline{[I^{**}]}, \, \bx'_m \in \cS_{\bu} }.
\]  
By Claim~\ref{claim:part weight}, each \( \by_{\bx'_m} \in \cS_{\bu}^* \) satisfies \( \wt_b(\by_{\bx'_m}) \leq t-s+h-1 \), i.e., \( \cS_{\bu}^* \subseteq \Ball_{t-s+h-1,b}(\mathbf{0}) \), and is also a \( (t-s-1, b) \)-burst-correcting code. Applying Theorem~\ref{thm:upper johnson}, we obtain  
\[
\abs*{\cS_{\bu}} = \abs*{\cS_{\bu}^*} \leq (t-s+h)(q^b-1)^h n^h \leq tq^{bh} n^h.
\]
 By summing up all $\abs{\cS_{\bu}}$, we have   
    \[
    \abs*{\cS} = \sum_{\bu \in \Sigma_q^\ell} \abs{\cS_{\bu}} \leq \sum_{\bu \in \Sigma_q^\ell} tq^{bh}n^{h} \leq t q^{b(h+10t)} n^h.
    \]
Since this holds for any transmitted codeword in $\cC_n$, we have
\[L(n, \cC_n, \ch(t,b), N(n)) \leq t q^{b(h+10t)} n^h.
\]
\end{proof}

\begin{cor}
\label{cor:upperh}
Let $q,b,t,s,h$ be constant positive integers, $q \geq 2$, and $t\geq s\geq h\geq 1$. Then, when using $(t-s-1,b)$-burst-correcting codes of length $n$, over $\ch(t,b)$, with $N(n)$ reads, $\Ord_b(N(n))=s-h$, the asymptotic growth rate of the list size satisfies
\[
    \La(h) \leq h.
\]
\end{cor}
\begin{proof}[Proof of~\cref{cor:upperh}]
Let \( \cC_n\subseteq \Sigma_q^n \) be a sequence of \( (t-s-1,b) \)-burst-correcting codes such that the list size \( L(n, \cC_n, \ch(t,b), N(n)) \) is maximized, i.e.,
\[
L(s, N(n)) = L(n, \cC_n, \ch(t,b), N(n)).
\]
Since $\Ord_b(N(n))=s-h$ implies $N(n)\geq \abs{\Ball_{s-h,b}(\mathbf{0})}+1$, we can use Theorem~\ref{thm:upper list recon} to get 
\[L(s, N(n)) = L(n, \cC_n, \ch(t,b), N(n)) \leq t q^{b(h+10t)} n^h.
\]
Finally, by definition,
\[
\La(h) = \inf_{N(n):\Ord_{b}(N(n))=s-h}\liminf_{n\to\infty} \frac{\log_q L(s,N(n))}{\log_q n} \leq h.
\]
\end{proof}

In summary, we obtain the following result:
\begin{theorem}\label{Theorem:Main 1:List}
Let $q,b,t,s,h$ be constant positive integers, $q \geq 2$, and $t\geq s\geq h\geq 1$. Then, when using $(t-s-1,b)$-burst-correcting codes of length $n$, over $\ch(t,b)$, with $N(n)$ reads, $\Ord_b(N(n))=s-h$, the asymptotic growth rate of the list size satisfies
\[
    \La(h) =  h.
\]
\end{theorem}
\begin{proof}[Proof of~\cref{Theorem:Main 1:List}]
Simply combine Theorem~\ref{thm:lower list recon} and Corollary~\ref{cor:upperh}.
\end{proof}
 
\section{List-Reconstruction Algorithm}
\label{sec:alg}

In Section~\ref{sec:list}, Theorem~\ref{thm:upper list recon}, we proved that any $(t-s-1,b)$-burst-correcting code, used over $\ch(t,b)$, with $N = \abs{\Ball_{s-h,b}(\mathbf{0})} + 1 = \Theta(n^{s-h})$ reads, has a list size at most $ tq^{2b(h+5t)}n^h$. There remains, of course, the problem of designing an efficient algorithm for finding this list. To this end, we leverage the \emph{majority logic with threshold} method, a powerful tool originally introduced by Levenshtein~\cite{levenshtein2001efficient}. This approach was initially developed for constructing algorithms in scenarios where the reconstruction code is the entire space. Over time, the technique has been adapted and extended in various contexts, as highlighted in recent works~\cite{abu2021levenshtein, junnila2023levenshtein}. Before describing the algorithm, we introduce some necessary definitions.

Let $Y = \set{\by_1,\ldots,\by_N} \subseteq \Sigma_q^n$ be a collection of sequences. For any $i\in [n]$ and $c\in \Sigma_q$, we define 
\[
m(j,c) = \abs*{\set{i\in [N]: \by_i[j] = c}},
\]
and 
\[
m(j) = \argmax_{c\in \Sigma_{q}} m(j,c).
\]
It is clear that $m(j)$ is unique when $\max_{c\in \Sigma_{q}} m(j,c) > \frac{N}{2}$. 

Next, we introduce the majority function $\Maj_\tau: 2^{\Sigma_q^n} \to {\Sigma_{*}^{n}}$, where we define $\Sigma_* = \Sigma_q \cup \set{*}$, and where $*$ denotes a new symbol not in $\Sigma_q$. For any $Y = \set{\by_1,\ldots,\by_N} \subseteq \Sigma_q^n$, and any given real $\tau\geq 0$, we define $\Maj_\tau(Y)$ to be the sequence in $\Sigma_*^n$ satisfying for all $j\in[n]$,
\[
\Maj_\tau(Y)[j] = \begin{cases}
    m(j), &\text{ if } \max_{c\in \Sigma_{q}} m(j,c) > \frac{N+\tau}{2},\\
    *, &\text{ otherwise.}
\end{cases}
\]

We present out list-reconstruction algorithm in Algorithm~\ref{alg:reconstruction}. While the code we use, $\cC$, is only a $(t-s-1,b)$-burst-correcting code, our algorithm uses as a black box a list-decoding algorithm for $\cC$ over $\ch(t-s+h,b)$, denoted $LD_{\cC}$. Namely, the algorithm $LD_{\cC}$ computes for any  $\by \in \Sigma_q^n$,
\[LD_{\cC}(\by)=\cC\cap \Ball_{t-s+h,b}(\by).\]
We denote the complexity of the list-decoding algorithm $LD_{\cC}$ as $\Com(LD_{\cC})$.

\begin{algorithm}
\caption{List reconstruction algorithm of a $(t-s-1,b)$-burst-correcting code $\cC\subseteq \Sigma_q^n$ over the channel $\ch(t,b)$ with $N= \abs{\Ball_{s-h,b}(\mathbf{0})} + 1$ reads}
\label{alg:reconstruction}
		\begin{algorithmic}[1]
			\Require{$N$ different sequences $Y = \set{\by_1,\ldots,\by_N} \subseteq \Ball_{t,b}(\bx)$ for some $\bx \in \cC$}                
			\Ensure{$\cC\cap\bigcap_{i\in[N]}\Ball_{t,b}(\by_i)$}
			\LineComment Set-up phase:\\
			$\cL^* \gets \emptyset$\\
                $\tau \gets (1 - \frac{2}{(t-s+h+1)b+1})N$\\
                $\bz \gets \Maj_\tau(Y)$\\
                $S \gets \set{j\in [n]: \bz[j] = *}$\\ $\cZ \gets \set{\bu \in \Sigma_q^n: \bu[j] = \bz[j], j\in [n]\setminus S}$
			
			\LineComment Traversal phase:
			
			\For {$\bu \in \cZ$} 
                \State $\cL_{\bu} \gets LD_{\cC}(\bu)$ \label{line:Lu}
                \State $\cL^* \gets \cL^* \cup \cL_{\bu}$
                \EndFor\\
                $\cL^* \gets \cL^*\cap \bigcap_{i=1}^N \Ball_{t,b}(\by_i)$ \label{line:int}
                \State \Return {$\cL^*$}
		\end{algorithmic}
	\end{algorithm}

\begin{theorem}\label{thm:alg}
Let $\cC \subseteq \Sigma_q^n$ be a $(t-s-1,b)$-burst-correcting code. Assume $N = \abs{\Ball_{s-h,b}(\mathbf{0})} + 1$, with $n \geq \max \set{4b(s-h), s^2 q^{3bt} (4s-4h)^{s-h} ((t-s+h+1)b+1)}$.
Then Algorithm~\ref{alg:reconstruction} list-reconstructs codewords of $\cC$ over $\ch(t,b)$, using $N$ reads, with complexity $O(\Com(LD_{\cC}) + n^{h+2}N)$.
\end{theorem}

\begin{proof}[Proof of~\cref{thm:alg}]
We begin by proving the correctness of the algorithm. Assume $\bc \in \cC$ was transmitted, and $Y = \set{\by_1,\ldots,\by_N} \subseteq \Ball_{t,b}(\bc)$ are $N$ distinct reads. Let $\cL=\set{\bx_1,\dots,\bx_L}$ be the true list, namely,
\[
\cL = \cC\cap\bigcap_{i\in[N]}\Ball_{t,b}(\by_i),
\]
and our goal is to prove that the algorithm's output, $\cL^*$, is the same, i.e., $\cL^*=\cL$.

By the algorithm, the set $\cZ$ contains exactly $q^{\abs{S}}$ sequences, which are all the ways of filling in the $*$ entries with elements from $\Sigma_q$. Let $\bu\in\cZ$, and denote $\bu[S]$ the restriction of $\bu$ to the coordinates in $S$.

\begin{claim}\label{claim:u-x}
If $\bu\in\cZ$, and $\bu[S]=\bx[S]$ for some $\bx\in\cL$, then $\bu \in \Ball_{t-s+h, b}(\bx)$.
\end{claim}
\begin{poc}
We first contend that the claim holds for $h=s$. Indeed, when $h=s$ we have $N=2$. If $\by_1,\by_2$ are the two received reads, they agree on all positions, except those in the set $S$. The sequence $\bu$, therefore, agrees with $\bx$ in the positions of $S$, and with $\by_1$ and $\by_2$ on the positions of $[n]\setminus S$. Thus,
\[
d_b(\bu,\bx) \leq d_b(\by_1,\bx) \leq t.
\]

Let us therefore consider the case of $h<s$. Assume to the contrary the claim does not hold. Then,
\[
k=\wt_b(\bu-\bx) \geq t-s+h+1.
\]
Hence, we can write $\bu-\bx=\be_1+\dots+\be_k$, where  $\be_1, \ldots, \be_k $ are disjoint $b^{\leq}$-bursts, each contained in disjoint cyclic intervals $ I_1, \ldots, I_{k} $, respectively. For the remainder of the proof, it suffices to focus on $t-s+h+1$ bursts of the $k$. Define $ I_i^* $ as the $ (1,b) $-extension of $ I_i $ for $ i \in [t-s+h+1] $, and let $ I^* = \bigcup_{i=1}^{t-s+h+1} I_i^* $. Denote  
\[
P = \set*{p_1, \ldots, p_m} = \bigcup_{i=1}^{t-s+h+1} \supp(\be_i) \subseteq \bigcup_{i=1}^{t-s+h+1} I_i \subseteq I^*,
\]  
where $ m \leq (t-s+h+1)b $, and $ P $ must be covered by at least $ t-s+h+1 $ intervals, each of length at most $ b $. Since $\supp(\bu - \bx) \subseteq [n] \setminus S $ and $\bu[j] = \bx[j]$ for $j\in S$, we have $ P \subseteq [n] \setminus S $, and $ \bu[j] = \bz[j] $ for all $ j \in [n]\setminus P $.  

For any $ Q \subseteq P $, define  
\[
Y_Q = \set*{\by_i \in Y : \text{for all } j \in P, \by_i[j] = \bz[j] \text{ if and only if } j \in Q},
\]  
and  
\[
S_{p_i} = \set*{Q \subseteq P : p_i \in Q}.
\]  

Then $ N = \abs{Y} = \sum_{Q \subseteq P} \abs{Y_Q} $, and for each $ i \in [m] $, we have  
\[
\sum_{Q \in S_{p_i}} \abs*{Y_Q} > \frac{N + \tau}{2},  
\]  
which follows from the definition of the majority function. On the one hand, we have 
    \[
    \sum_{i=1}^m \sum_{Q\in S_{p_i}} \abs*{Y_Q} > \frac{N+\tau}{2}m.
    \]
On the other hand, we have 
    \[
     \sum_{i=1}^m \sum_{Q\in S_{p_i}} \abs*{Y_Q} = \sum_{Q\subseteq P} \abs*{Q} \abs*{Y_Q} \leq \sum_{Q\subseteq P} (m-1)\abs*{Y_Q} + \abs*{Y_P} = (m-1)N + \abs{Y_P}.
    \]
    Then $(m-1)N + \abs{Y_P} > \frac{N+\tau}{2}m$. By simplifying the inequality, we have 
    \begin{equation}\label{eq:burst count}
    \tau < N - \frac{2}{m}(N-\abs{Y_P}) \leq N - \frac{2}{(t-s+h+1)b}(N-\abs{Y_P}). 
    \end{equation}

We now compute the size of $Y_P$. For all $i\in[N]$, let $\balpha_i = \by_i - \bx\in B_{\leq t, \leq b}$ be the error of $\by_i$ w.r.t.~$\bx$. Since $(\bz - \bx)[j] = (\bu - \bx)[j] \neq 0$ for any $j\in P$, we have
\[
Y_P = \set*{\by_i \in Y:  \by_i[j] = \bz[j] \text{ for } j\in P}
\subseteq \set*{\bx+\balpha : \balpha \in B_{\leq t,\leq b}, \balpha [j] = \bz[j]-\bx[j] \neq 0 \text{ for } j\in P}.
\]
Hence,
\[
\abs*{Y_P} \leq \sum_{\ell=t-s+h+1}^t \abs*{\set*{\balpha \in B_{\ell,\leq b}: \balpha[j] = (\bz -\bx)[j]\neq 0, \text{ for } j\in P}}.
\]
    
Since $P \subseteq I \subseteq I^*$, we have for each $\balpha\in B_{\ell,\leq b}$ as in the expression above, positions $[n] \setminus I^*$ contain at most $\ell-(t-s+h+1)$ number of $b^{\leq}$-bursts. In other words, $\balpha\overline{[I^*]} \in B_{\leq \ell-t+s-h-1,\leq b}(n-\abs{I^*},q)$. By Remark~\ref{rmk:ballsizecomb} and the fact that $n\geq 4b(s-h)$ as well as $\abs{I^*} \leq 3b(t-s+h+1)$, we have 
\begin{align*}
&\abs*{\set*{\balpha \in B_{\ell,\leq b}: \balpha[j] = (\bz -\bx)[j]\neq 0, \text{ for } j\in P}}
\leq \abs*{B_{\leq \ell-t+s-h-1,\leq b}(n-\abs{I^*},q)} q^{\abs{I^*}} \\
&\leq \abs*{B_{\leq s-h-1,\leq b}(n,q)} q^{\abs{I^*}} 
\leq (s-h)\binom{n}{s-h-1}(q^b-1)^{s-h-1} q^{3b(t-s+h+1)} \\
&\leq (s-h)q^{3bt} \binom{n}{s-h-1}.
\end{align*}
If follows that
\[
\abs*{Y_p} \leq (s-h)^2 q^{3bt} \binom{n}{s-h-1} \leq s^2 q^{3bt} n^{s-h-1}.
\]
    
Since $n \geq 4b(s-h)$, by Remark~\ref{rmk:ballsizecomb}, we have 
\[
N=\abs*{\Ball_{s-h,b}(\mathbf{0})} + 1 > \parenv*{1-\frac{1}{q}}^{s-h}(q^b-1)^{s-h} \binom{n-(2b-2)(s-h)}{s-h} \geq \parenv*{\frac{n}{4(s-h)}}^{s-h}.
\]
Thus,
\[
    \abs*{Y_p} < \frac{N}{(t-s+h+1)b+1},
\]
for $n \geq \max \set{4b(s-h), s^2 q^{3bt} (4s-4h)^{s-h} ((t-s+h+1)b+1)}$. 
Therefore, Equation~\eqref{eq:burst count} leads to a contradiction since
    \[
    \tau = \parenv*{1 - \frac{2}{(r-s+h+1)b+1}}N < N - \frac{2}{(t-s+h+1)b} \parenv*{N - \frac{N}{(t-s+h+1)b+1}} = \tau.
    \]
\end{poc}

We proceed to show that $\cL$ and $\cL^*$ are subsets of each other, and therefore equal. In the first direction, by assumption, $LD_{\cC}(\bu)\subseteq \cC$, for every $\bu\in\Sigma_q^n$, and thus, $\cL^*\subseteq\cC$. Additionally, by line~\ref{line:int}, we have $\cL^*\subseteq\bigcap_{i\in[N]}\Ball_{t,b}(\by_i)$. Thus,
\[
\cL^*\subseteq \cC\cap\bigcap_{i\in[N]}\Ball_{t,b}(\by_i)=\cL.
\]
In the other direction, let $\bx\in\cL$, and let $\bu\in\cZ$ be such that $\bu[S]=\bx[S]$. By Claim~\ref{claim:u-x}, $d_b(\bu,\bx)\leq t-s+h$. Thus, in line~\ref{line:Lu}, $\bx\in\cL_{\bu}$, and $\bx$ does not get removed in line~\ref{line:int} since $\bx\in\cL\subseteq \bigcap_{i\in[N]}\Ball_{t,b}(\by_i)$. Therefore, when the algorithm finishes, $\bx\in\cL^*$. This gives us
\[
\cL \subseteq \cL^*,
\]
and hence, $\cL^*=\cL$, namely, the algorithm computes the correct list.
    
To compute the complexity of the algorithm, we first show that the size of $S$ is a constant. Let $A$ be an $N\times n$ array such that the $i$-th row $A[i]$ is $\balpha_i = \by_i - \bc$, where $\bc\in\cC$ is the transmitted codeword. Since $\balpha_i \in B_{\leq t,\leq b}$, the total number of non-zero entries of $A$ is at most $btN$. On the other hand, for each column $j\in S$, the number of non-zero entries in $j$-th column is 
\[
N - m(j,\bc[j]) \geq N - \max_{c\in \Sigma_q}\set*{m(j,c)} \geq N - \frac{N+\tau}{2} = \frac{N}{(t-s+h+1)b+1}.
\]
Thus, the total number of non-zero entries of $A$ is at least $\frac{N}{(t-s+h+1)b+1} \abs{S}$. Then 
\[
\abs{S} \leq bt((t-s+h+1)b+1) = O(1).
\]

Next, we bound the list size $\cL_{\bu}$ computed in line~\ref{line:Lu}. Since $\cC$ is a $(t-s-1,b)$-burst-correcting code, and the decoding radius is $t-s+h$, the list size $\abs{\cL_{\bu}}$ from line~\ref{line:Lu}, is upper bounded by the number of codewords a $(t-s-1,b)$-burst-correcting code may have inside a ball of radius $t-s+h$. By Theorem~\ref{thm:upper johnson},
\[
\abs*{\cL_{\bu}} \leq (t-s+h+1)(q^b-1)^{h+1}n^{h+1}=O(n^{h+1}).
\]
Since we have $O(1)$ iterations, $\abs{\cL^*}=O(n^{h+1})$ as well.

In the Set-up phase, the time complexity is $O(n N)$. For the traversal phase, we have $O(1)$ iterations of running $LD_{\cC}$, requiring $O(\Com(LD_{\cC}))$. Taking the union of lists of length $O(n^{h+1})$ takes $O(n^{h+1}\log n)$. Finally, taking the intersection in line~\ref{line:int} requires $N$ distance checks for each element in $\cL^*$. Measuring the burst distance between two vectors of length $n$ takes $O(n)$ times (see the appendix). Thus, line~\ref{line:int} has time complexity $O(n^{h+2}N)$. Therefore the total time complexity is $O(\Com(LD_{\cC}) + n^{h+2}N)$.
\end{proof}
 
\section{Conclusion}\label{sec:con}
In this paper, we explored the trade-off between the redundancy of reconstruction codes and the number of reads in a channel subject to at most $t$ burst substitutions of length $\leq b$. We established bounds on the redundancy of burst-correcting codes, which can be viewed as a special case of reconstruction codes with a single read. For cases involving $\Theta(n^s)$ reads, the optimal choice for reconstructing codewords is a $t-s-1$ error-correcting code (as discussed in Section~\ref{sec:reconstruction code}). Our findings generalize the work of Levenshtein~\cite{levenshtein2001efficient}, extending the analysis from channels with multiple substitutions to those with multiple burst substitutions.  For scenarios involving smaller $\Theta(n^{s-h})$ reads,  $t-s-1$ error-correcting codes can be list-decoded with list size $O(n^h)$ has been studied in Section~\ref{sec:list}. Roughly speaking, we can summarize this in the following way: In the presence of a channel introducing up to $t$ burst errors of length at most $b$, we have a ``budget'' of $t-1$ which we can divide between error-correction capability $\epsilon$, number of reads $\rho$, and list size $\lambda$,
\[
t-1 = \epsilon + \rho + \lambda.
\]
Then we can use, over this channel, an burst-correcting code capable of correcting $\epsilon$ errors, with $\Theta(n^\rho)$ reads, and list size $O(n^\lambda)$, where $\epsilon,\rho,\lambda$ are non-negative integers.

Several open questions remain. One area of interest is the development of error-correcting codes for channels with multiple burst substitutions that offer improved redundancy compared to the improved GV bound. While it is known that $\ra(\ch(t,b),s)=\ra(\ch(t-s-1,b),-1)$ (see Theorem~\ref{thm:MARR for burst}), the exact values have not yet been fully determined. Finally, we also leave Conjecture~\ref{conj: burst} as a topic for future investigation.

\appendix 

\section{Appendix -- Computing the Burst Distance}

In this appendix we present a straightforward linear-time algorithm for finding $d_b(\bx,\by)$ for any two $\bx,\by\in\Sigma_q^n$. The algorithm is given in Algorithm~\ref{alg:distance}.

To prove correctness, assume $ d_b(\bx, \by) = k $, so $ \supp(\be) $ is covered by $ k $ disjoint cyclic intervals $ I_1, \ldots, I_k $. For each $ p \in P $, the algorithm greedily generates $ t_p $ cyclic intervals $ J_1, \ldots, J_{t_p} $ covering $ \supp(\be) $. Hence, the output satisfies $ t_p \geq k $. Now assume $ p_1 \in I_1 $. The starting point $ f_1 $ of $ I_1 $ must belong to $ P $. For $ p = f_1 $, the algorithm constructs $ t_{f_1} $ intervals $ J_1, \ldots, J_{t_{f_1}} $ that cover $\supp(\be)$. It is easy to see that once the right starting point is used, a greedy choice of intervals results in $t_{f_1}\leq k$, and therefore, $ t_{f_1} = k $. The algorithm outputs $ \min t_p = k$ ensuring the correctness.

\begin{algorithm}
    \caption{Calculation of the $d_b$ burst distance}
    \label{alg:distance}
    \begin{algorithmic}[1]
        \Require {$\bx, \by \in \Sigma_q^n$, maximum burst length $b$}
        \Ensure {$d_b(\bx, \by)$}

        \If{$\bx=\by$}
            \State \Return $0$
        \EndIf
        \State $\be \gets \bx - \by$ \Comment{Compute the error vector}
        \State $p_1 \gets \min \supp(\be)$ \Comment{Find leftmost error position}
        \State $P \gets [p_1 - b + 1, p_1] \cap \supp(\be)$ 
        \State $T_P\gets \emptyset$
        
        \For{$p \in P$}
            \State $U \gets \supp(\be)$ \Comment{Remaining error positions}
            \State $t_p \gets 0$ \Comment{Burst count for current $p$}
            
            \While{$U \neq \emptyset$}
                \State $I_{t_p} \gets [p, p + b - 1]$
                \State $U \gets U \setminus I_{t_p}$
                \State $p \gets \min U$ \Comment{Next starting point}
                \State $t_p \gets t_p + 1$
            \EndWhile
            \State $T_P\gets T_P\cup\set{t_p}$
        \EndFor
        \State \Return $\min T_p$ \Comment{Minimum burst count across all start points}
    \end{algorithmic}
\end{algorithm}

\bibliographystyle{abbrv}
\bibliography{reference}

\end{document}